%% file: gse_2b_3b.tex
\documentclass[a4paper,11pt,reqno]{amsart}

\usepackage[margin=35mm]{geometry}
\usepackage[T1]{fontenc}
\usepackage[english]{babel}
\usepackage[utf8]{inputenc}
\usepackage{amsmath,amssymb,amsthm,amsfonts,amstext,amsopn,amsxtra,dsfont,mathrsfs,esint,enumitem,bookmark,mathtools}
\usepackage{bbm}
\usepackage[capitalise]{cleveref}
\hypersetup{linktocpage=true}
\usepackage{microtype}
\usepackage[headings]{fullpage}\usepackage{microtype}
\usepackage[headings]{fullpage}

\newcommand\numberthis{\addtocounter{equation}{1}\tag{\theequation}}

\theoremstyle{plain}
\newtheorem{theorem}{Theorem}
\newtheorem{lemma}[theorem]{Lemma}
\newtheorem{proposition}[theorem]{Proposition}

\theoremstyle{definition}

\newtheorem{assumption}[theorem]{Assumption}

\newcommand\xqed[1]{%
	\leavevmode\unskip\penalty9999 \hbox{}\nobreak\hfill\quad\hbox{#1}%
}
\newcommand\remarkend{\xqed{$\triangle$}}

\makeatletter
\def\@endtheorem{\remarkend\endtrivlist\@endpefalse }
\makeatother
\theoremstyle{remark}

\makeatletter
\def\@endtheorem{\endtrivlist\@endpefalse }
\makeatother

\crefname{theorem}{Theorem}{Theorems}
\crefname{lemma}{Lemma}{Lemmas}
\crefname{proposition}{Proposition}{Propositions}
\crefname{corollary}{Corollary}{Corollaries}
\crefname{definition}{Definition}{Definitions}
\crefname{assumption}{Assumption}{Assumptions}
\crefname{remark}{Remark}{Remarks}

\crefname{subsection}{subsection}{subsections}
\crefname{subsubsection}{subsection}{subsections}

\renewcommand{\d}[1]{\ensuremath{\operatorname{d}\!{#1}}}

\newcommand{\Tr}{\operatorname{Tr}}

\newcommand{\Supp}{\operatorname{Supp}}

\newcommand{\diver}{\operatorname{div}}

\newcommand{\cT}{\mathcal{T}}

\newcommand{\cM}{\mathcal{M}}

\newcommand{\cK}{\mathcal{K}}

\newcommand{\cI}{\mathcal{I}}
\newcommand{\cJ}{\mathcal{J}}

\newcommand{\R}{\mathbb{R}}

\newcommand{\s}{\textmd{\normalfont s}}

\newcommand{\ceqq}{\coloneqq}

\newcommand\mydots{\ifmmode\mathellipsis\else.\kern-0.08em.\kern-0.08em.\fi}

\allowdisplaybreaks

\DeclareUnicodeCharacter{2217}{*}

\usepackage{csquotes}
\usepackage[backend=bibtex,sorting=nty,citestyle=numeric-comp,bibstyle=ieee,maxbibnames=99,doi=false,isbn=false,url=false,eprint=false,dashed=true]{biblatex}
\usepackage{graphicx}
\usepackage{subcaption}
\usepackage[table,xcdraw]{xcolor}

\usepackage{tikz}
\usetikzlibrary{arrows}
\usepackage{tkz-tab}

\usepackage{standalone}

\usepackage{hyperref}
\hypersetup{bookmarksdepth=3}
\title[Ground state energy of dilute Bose gases]{Ground state energy of dilute Bose gases with two-body and three-body interactions}
\author{François L. A. Visconti}
\address{Department of Mathematics, LMU Munich, Theresienstrasse 39, 80333 Munich, Germany}
\email{visconti@math.lmu.de}

\begin{document}
	\maketitle
	
	\section*{Abstract}
	We study dilute Bose gases in the thermodynamic limit interacting via two-body and three-body interaction potentials. We prove that the leading order of the thermodynamic ground state energy is entirely characterised by both the scattering length of the two-body potential and the scattering energy of the three-body potential. The corresponding result for two-body interactions was proven in seminal papers of Dyson (1957) \cite{Dyson1957gseHS} and Lieb--Yngvason (1998) \cite{Lieb1998GSE}, and the result for three-body interactions was proven much more recently by Nam--Ricaud--Triay (2022) \cite{Nam2022ground}. The present result resolves a conjecture of Nam--Ricaud--Triay (2022) \cite{Nam2022dilute}.
		
	\tableofcontents
	
	\section{Introduction}\label{sec1}
	
	In the study of dilute Bose gases, where particle collisions occur rarely, interactions between particles are often modelled through an effective two-body potential $V(x_1-x_2)$. However, this approximation cannot explain certain physical properties of the Bose gas \cite{surface_liquid_water_three_body,three_body_interactions_condensed_phases_helium} and a three-body correction $W(x_1-x_2,x_1-x_3)$ must therefore be added. Moreover, dilute Bose gases with attractive two-body interactions and repulsive three-body ones have gained considerable interest since three-body interactions can stabilise the condensate against a collapse induced by two-body interactions \cite{Gammal2000atomicBEC,Koch2008stabilisation,petrov2023beyondmeanfield}.
	
	In the case of Bose gases interacting only via a two-body potential $V$, it was proven in \cite{Dyson1957gseHS} (upper bound) and in \cite{Lieb1998GSE} (lower bound) that the ground state energy per unit volume in the thermodynamic limit satisfies
	\begin{equation}
		\label{eq:ground_state_energy_two_body}
		e_{\textmd{2B}}(\rho, V) = 4\pi a(V)\rho^2(1 + o(1))
	\end{equation}
	as $\rho a(V)^3\rightarrow0$. Here, $\rho$ designates the density of the system and $a(V)$ is the scattering length of the two-body potential; the limit $\rho a(V)^3\rightarrow 0$ corresponds to the dilute regime. See \cite{Basti2021secondOU,Fournais2020energyDB,Fournais2023energyDB2,Yau2009secondOU} for rigorous results on the next order corrections to \eqref{eq:ground_state_energy_two_body}.
	
	For systems interacting only via a three-body potential $W$, it was recently proven in \cite{Nam2022ground} that the thermodynamic ground state energy per unit volume satisfies
	\begin{equation}
		\label{eq:ground_state_energy_three_body}
		e_{\textmd{3B}}(\rho,W) = \dfrac{1}{6}b_\cM(W)\rho^3(1 + o(1))
	\end{equation}
	as $\rho b_\cM(W)^{3/4}\rightarrow 0$, where $\rho$ denotes the density of the system and $b_\cM(W)$ is the scattering energy associated to $W$; the limit $\rho b_\cM(W)^{3/4} \rightarrow0$ again corresponds to the dilute regime. The time dependent problem with the mean-field type potential $N^{6\beta - 2}W(N^\beta(x-y,x-z))$ for $\beta \geq 0$ small has already been studied in \cite{Adami2023microscopicDS,Chen2011quinticNLS,Chen2012secondOC,Chen2018TheDO,lee2021rateCT,Li2021derivationNS,Nam2019derivation3D,Rout2024microscopicDG,Yuan2015derivationQNLS}, the leading order of the ground state energy in the Gross--Pitaevskii limit $\beta=1/2$ was derived in \cite{Nam2023condensation}, and the aforementioned stabilisation of the condensate against collapse has been studied in \cite{Nguyen2023onedimensional,Nguyen2023stabilization}, in one and two dimensions, for soft scaling potentials.
	
	Though the works of Nam--Ricaud--Triay \cite{Nam2022ground,Nam2023condensation} focused on three-body interactions alone, the more physically relevant case is the one of combined two-body and three-body interactions, which is the focus of the present paper. More specifically, we derive the leading order of the thermodynamic ground state energy per unit volume of a dilute Bose gas interacting with combined two-body and three-body interactions. This generalises \eqref{eq:ground_state_energy_two_body} and \eqref{eq:ground_state_energy_three_body}, and proves \cite[Conjecture 7]{Nam2022dilute}.
	
	\subsection{Model presentation}
	
	We consider a system of $N$ bosons trapped in a box $\Lambda_L\coloneqq\left[-L/2,L/2\right]^3$ of side length $L > 0$ interacting with a two-body potential $V$ and a three-body potential $W$. The energy is described by the $N$-body Hamiltonian
	\begin{equation}
		\label{eq:hamiltonian}
		\displaystyle H_{N,L} = \sum_{i=1}^N-\Delta_i + \sum_{1\leq i<j\leq N}V(x_i - x_j) + \sum_{1\leq i<j<k\leq N}W(x_i-x_j,x_i-x_k)
	\end{equation}
	acting on the bosonic space $L^2_\textmd{s}(\Lambda_L^N) \coloneqq \bigotimes_{\textmd{sym}}^NL^2(\Lambda_L)$, where $-\Delta$ denotes the Laplacian with Neumann boundary condition on $\Lambda_L$ ($\nabla f\cdot\overrightarrow{n}=0$ on $\partial\Lambda_L$). Moreover, we consider interaction potentials satisfying the following assumptions.
	\begin{assumption}[Potentials]
		\label{assumption:potentials}
		The two-body potential $0\leq V\in L^1(\R^3)$ is radial with compact support. The three-body potential $0 \leq W \in L^1(\R^6)$ has compact support and satisfies the \textit{three-body symmetry} properties
		\begin{equation}
			\label{eq:three_body_symmetry}
			W(x,y) = W(y,x)\quad\text{and}\quad W(x-y,x-z) = W(y-x,y-z) = W(z-y,z-x).
		\end{equation}
		We fix $R_0 > 0$ such that $\Supp V\subset B(0,R_0)$ and $\Supp W\subset B(0,R_0)$.
	\end{assumption}
	
	Under these assumptions, $H_{N,L}$ can be defined as a positive
	self-adjoint operator by Friedrichs’ method, and it has compact resolvent.
	
	The thermodynamic ground state energy per unit volume $e(\rho,V,W)$ is defined as
	\begin{equation}
		\label{eq:thermodynamic_ground_state_energy}
		e(\rho,V,W) = \lim_{\substack{N\rightarrow\infty\\N/L^3\rightarrow\rho}}\inf_{\|\Psi\|=1}\dfrac{\left<\Psi,H_{N,L}\Psi\right>}{L^3}.
	\end{equation}
	This limit exists and is independent of boundary conditions; see for instance \cite{statistical_mechanics}. We will estimate $e(\rho,V,W)$ in terms of the scattering length of $V$ and the scattering energy of $W$.
	
	The scattering length $a(V)$ of $V$ is defined by
	\begin{equation*}
		8\pi a(V) = \inf_{g\in D^1(\mathbb{R}^3)} \int_{\mathbb{R}^3}\d{}x\left(2\left|\nabla g(x)\right|^2 + V(x)\left|1 - g(x)\right|^2\right),
	\end{equation*}
	where $D^1(\mathbb{R}^d)$ denotes the space of functions $g:\mathbb{R}^d \rightarrow \mathbb{C}$ in $L_{\textmd{loc}}^1(\mathbb{R}^d)$ vanishing at infinity and satisfying $\left|\nabla g\right|\in L^2(\mathbb{R}^d)$ (see \cite[Section 8.3]{lieb_loss}). In contrast, the \textit{three-body scattering energy} $b_\cM(W)$ of $W$ is defined by
	\begin{equation*}
		b_{\mathcal{M}}(W) = \inf_{g\in D^1(\mathbb{R}^6)} \int_{\mathbb{R}^6}\d{}\mathbf{x}\left(2\left|\mathcal{M}\nabla g(\mathbf{x})\right|^2 + W(\mathbf{x})\left|1 - g(\mathbf{x})\right|^2\right),
	\end{equation*}
	where the matrix $\mathcal{M}:\mathbb{R}^3\times\mathbb{R}^3\rightarrow\mathbb{R}^3\times\mathbb{R}^3$ is given by
	\begin{equation}
		\label{eq:three_body_symmetry_matrix_intro}
		\mathcal{M} =
		\dfrac{1}{2\sqrt{2}}
		\begin{pmatrix}
			\sqrt{3} + 1 & \sqrt{3} - 1\\
			\sqrt{3} - 1 & \sqrt{3} + 1
		\end{pmatrix}.
	\end{equation}
	The matrix $\cM$ arises naturally when rewriting the three-body scattering problem associated to $W$ in terms of relative coordinates (see Section~\ref{section:modified_scattering_energy}). We refer to \cite[Section 2.2]{Nam2023condensation} for a more exhaustive discussion on the matter
	
	\subsection{Dilute regime}
	
	It is well know that the length scale at which a two-body potential $V$ acts is characterised by its scattering length $a(V)$. Hence, for particles interacting only via two-body interactions, the dilute regime can be encoded in the limit
	\begin{equation}
		\label{eq:dilute_regime_two_body}
		\rho a(V)^3 \rightarrow 0,
	\end{equation}
	which expresses that the effective range of the interaction ($\sim a(V)$) is much smaller than the mean distance between particles ($\sim\rho^{-1/3}$). For systems interacting via three-body interactions, it was shown in \cite{Nam2022ground,Nam2023condensation} that $b_\cM(W)^{1/4}$ plays the same role as the scattering length. The dilute regime thus corresponds to the limit
	\begin{equation}
		\label{eq:dilute_regime_three_body}
		\rho b_\mathcal{M}(W)^{3/4}\rightarrow 0.
	\end{equation}
	
	Therefore, for particles interacting via both a two-body potential $V$ and a three-body potential $W$, the dilute regime can be encoded in the two limits \eqref{eq:dilute_regime_three_body} and \eqref{eq:dilute_regime_two_body}, which can be combined as
	\begin{equation}
		\label{eq:dilute_regime_intro}
		\rho \mathfrak{a}^3 \rightarrow 0,
	\end{equation}
	with $\mathfrak{a} = \max\{a(v),\rho b_\mathcal{M}(w)\}$. The parameter $\mathfrak{a}$ plays the role of an effective combined scattering length; we explain this in more detail in Section~\ref{section:ingredients_proof}.

	\subsection{Main result}
	
	Our main result is the following.
	\begin{theorem}
		\label{th:main_result}
		\sloppy Let $V$ and $W$ satisfy Assumption~\ref{assumption:potentials}. Define $Y = \rho\mathfrak{a}^3$ with $\mathfrak{a} = \max(a(V),\rho b_\mathcal{M}(W))$. Then, under the assumption that $R_0\mathfrak{a}^{-1}$ is bounded, the thermodynamic ground state energy per unit volume defined in \eqref{eq:thermodynamic_ground_state_energy} satisfies
		\begin{equation}
			\label{eq:thermodynamic_ground_state_energy_first_order}
			e(\rho,V,W) = \left(4\pi a(V)\rho^2 + \dfrac{1}{6}b_\mathcal{M}(W)\rho^3\right)\left(1 + \mathcal{O}(Y^\nu)\right),
		\end{equation}
		in the dilute limit $Y\rightarrow 0$, for some universal constant $\nu>0$ (independent of $V$ and $W$).
	\end{theorem}
	
	This result justifies \cite[Conjecture 7]{Nam2022dilute}. To prove Theorem~\ref{th:main_result}, we prove a lower bound (see Section~\ref{section:proof_energy_lower_bound}) and a matching upper bound (see Section~\ref{section:energy_upper_bound}). The main difficulty in deriving \eqref{eq:thermodynamic_ground_state_energy_first_order} is to extract the correct correlation energy simultaneously from the two-body and three-body potentials. Here are some remarks on our result.
	\begin{enumerate}
		\item Though the simple choice $\rho \rightarrow 0$ might seem more natural, the dilute regime encapsulated in the limit $Y\rightarrow0$ is subtler and important when considering combined two-body and three-body interactions. Indeed, in the limit $\rho \rightarrow 0$, $4\pi a(V)\rho^2$ and $b_\mathcal{M}(W)\rho^3$ are obviously not of the same order, whereas they can be in the dilute regime $Y\rightarrow0$. A possible way to achieve the latter limit is to fix the potentials $V$ and $W$, and to consider the rescaled potentials $V_\alpha \coloneqq \alpha^{-2}V(\alpha^{-1}\cdot)$ and $W_\delta \coloneqq \delta^{-2}W(\delta^{-1}\cdot)$ for some parameters $\alpha,\delta > 0$. Then, $a(V_\alpha) = \alpha a(V)$ and $b_\mathcal{M}(W_\delta) = \delta^4b_\cM(W)$, and the limit simply corresponds to $\rho \alpha^3,\rho\delta^3 \rightarrow 0$.
		\item A simple adaptation of our proof yields the leading order of the ground state energy in the Gross--Pitaevskii limit. More specifically, for the Hamiltonian
		\begin{equation*}
			\tilde{H}_N = \sum_{i=1}^{N}-\Delta_i + \sum_{1\leq i<j\leq N}N^2V(N(x_i-x_j)) + \sum_{1\leq i<j<k\leq N}NW(N^{1/2}(x_i-x_j,x_i-x_k)),
		\end{equation*}
		defined on $L^2_\s\left(\Lambda_1^N\right)$, where $V$ and $W$ satisfy Assumption~\ref{assumption:potentials}, one can prove that
		\begin{equation*}
			E_N^\textmd{GP} = \left(4\pi a(V) + \dfrac{b_\mathcal{M}(W)}{6}\right)N + o(N)
		\end{equation*}
		\sloppy as $N\rightarrow \infty$, where $E_N^\textmd{GP}$ denotes the ground state energy of $\tilde{H}_N$ on $L^2_\textmd{s}(\Lambda_1^N)$.
		\item Our proof can easily be extended to hard-core interactions and assuming that $R_0\mathfrak{a}^{-1} \leq C$, for some constant $C$, the error depends only on $C$ and $Y^\nu$. In older versions of the present paper \footnote{Available at \href{https://arxiv.org/abs/2402.05646v2}{arXiv:2402.05646v2}.}, the upper bound was derived using a localisation argument and second quantisation techniques as in \cite{Nam2022ground}, and did not cover the hard-core case.
		\item For Bose gases with two-body interactions, Lee-Huang-Yang \cite{Lee1957eigenvaluesEB} predicted the expansion
		\begin{equation*}
			e_\textmd{2B}(\rho,V) = 4\pi a(V)\rho^2\left[1 + \dfrac{128}{15\sqrt{\pi}}\left(\rho a(V)^3\right)^{1/2} + \cdots\right]
		\end{equation*}
		as $\rho a(V)^3\rightarrow 0$. See \cite{Basti2021secondOU,Fournais2020energyDB,Fournais2023energyDB2,Yau2009secondOU} for rigorous results, \cite{Haberberger2023freeED} for a result at positive temperature and \cite{Fournais2022groundSE} for a related result in the 2D case. At first sight, one might think that since the second order of this expansion is proportional to $\rho^{5/2}$, it dominates the three-body term in \eqref{eq:thermodynamic_ground_state_energy_first_order} as it is proportional to $\rho^3$. This is however not so simple since the dilute regime we consider is more subtle than the simple limit $\rho \rightarrow 0$.
		\item As shown in \eqref{eq:thermodynamic_ground_state_energy_first_order}, the leading order of the ground state energy is entirely characterised by the respective scattering problems of $V$ and $W$, meaning that correlation between two-body and three-body interactions does not play a role. Heuristically, this decoupling can be explained by the different length scales at which the two potentials act.
	\end{enumerate}
	
	\subsection{Strategy of the proof}
	
	\label{section:ingredients_proof}
	
	An important length scale is the Gross--Pitaevskii (GP) length scale $\ell_\textmd{GP}$, at which the spectral gap of the kinetic operator is of the same order of magnitude as the interaction energy per particle. For a system of $n$ particles trapped in a box $\Lambda_\ell$ that interact via a two-body interaction potential $V$, this translates to
	\begin{equation*}
		na(\ell^2V(\ell\cdot)) \simeq 1 \Longleftrightarrow \ell \simeq \dfrac{1}{\sqrt{\rho a(V)}} = \dfrac{a(V)}{\sqrt{\rho a(V)^3}},
	\end{equation*}
	where we used the scaling property $a(\ell^2V(\ell\cdot)) = \ell^{-1}a(V)$, and that the density is $\rho = n\ell^{-3}$. For three-body interactions however, the Gross--Pitaevskii length scale corresponds to
	\begin{equation*}
		n^2b_\mathcal{M}(\ell^2W(\ell\cdot)) \simeq 1 \Longleftrightarrow \ell \simeq \dfrac{1}{\rho b_{\mathcal{M}}(W)^{1/2}} = \dfrac{b_\mathcal{M}(W)^{1/4}}{\rho b_\mathcal{M}(W)^{3/4}},
	\end{equation*}
	where we again used the scaling property $b_\mathcal{M}(\ell^2W(\ell\cdot)) = \ell^{-4}b_\mathcal{M}(W)$. In the combined case, we thus expect the spectral gap of the kinetic operator to be of the same order of magnitude as both the two-body contribution and the three-body one when both Gross--Pitaevskii length scales coincide, meaning that
	\begin{equation}
		\label{eq:scattering_energies_critical_cases}
		a(v) \simeq \rho b_\mathcal{M}(W).
	\end{equation}
	Note that this condition ensures that the leading order of the two-body and three-body energies - which are respectively of order $a(V)\rho^2$ and $b_\mathcal{M}(W)\rho^3$ (see \eqref{eq:ground_state_energy_two_body} and \eqref{eq:ground_state_energy_three_body})
	- are comparable. In order to properly encapsulate this, we therefore consider the combined Gross--Pitaevskii length scale
	\begin{equation}
		\label{eq:combined_gp_length_scale}
		\ell_\textmd{GP} \simeq \min\left(\dfrac{1}{\sqrt{\rho a(V)}},\dfrac{1}{\rho b_\mathcal{M}(W)^{1/2}}\right).
	\end{equation}
	The idea being that the spectral gap of the kinetic operator thus controls the interaction energies per particle of both potentials. In terms of the effective combined scattering length $\mathfrak{a} = \max(a(V),\rho b_\mathcal{M}(W))$ (note that $\rho b_\mathcal{M}(W)^{3/4}$ is dimensionless), the definition \eqref{eq:combined_gp_length_scale} rewrites as
	\begin{equation*}
		\ell_\textmd{GP} \simeq \dfrac{\mathfrak{a}}{\sqrt{\rho \mathfrak{a}^3}}.
	\end{equation*}
	\\
	
	\noindent
	\textit{Lower bound.} We combine the strategies of \cite{Lieb1998GSE,Nam2022ground,Nam2016large}. Namely, we divide the box $\Lambda_L = \left[-L/2,L/2\right]^3$ into smaller boxes of side length $\ell > 0$ and estimate the energy separately in each box. Each of them thus contains $n$ particles described by the Hamiltonian $H_{n,\ell}$ defined in \eqref{eq:hamiltonian}. By dilation, this is equivalent to considering the rescaled Hamiltonian
	\begin{equation*}
		\tilde{H}_{n,\ell} = \sum_{i=1}^n-\Delta_i + \sum_{1\leq i<j\leq n}\ell^2V(\ell(x_i-x_j)) + \sum_{1\leq i<j<k\leq n}\ell^2W(\ell(x_i-x_j,x_i-x_k))
	\end{equation*}
	\sloppy on $L^2_\textmd{s}(\Lambda_1^{n})$. Indeed, defining the unitary operator $\mathcal{U}:L^2_\textmd{s}(\Lambda_\ell^{n}) \rightarrow L^2_\textmd{s}(\Lambda_1^{n})$ that acts as $\mathcal{U}\Psi = \ell^{3n/2}\Psi(\ell\cdot)$, we have $H_{n,\ell} = \ell^{-2}\mathcal{U}^*\tilde{H}_{n,\ell}\mathcal{U}$ in the quadratic form sense on $L^2_\s(\Lambda_\ell^n)$.
	
	Since we consider nonnegative potentials, we can simply neglect the interaction between boxes without increasing the energy. Moreover, we use a convexity argument to control the number of particles in each box. We consider boxes with a side length much shorter than the Gross--Pitaevskii length scale: $\ell \ll \ell_\textmd{GP}$. At such length scales, the spectral gap of the kinetic operator is large enough for the interaction potentials to be treated as perturbations of the kinetic operator using Temple's inequality. However, we do not do so directly since naive perturbation theory is expected not to give a good result for singular potentials (see e.g. \cite{Lieb2005MathematicsBG}). To circumvent this, we first replace the interaction potentials by softer ones using a version of Dyson's lemma adapted to combined two-body and three-body interactions.
	
	Let us briefly explain our approach. The original idea of the two-body Dyson lemma \cite{Dyson1957gseHS,Lieb1998GSE} is to sacrifice some of the kinetic energy in order to bound the singular potential $\ell^2V(\ell\cdot)$ from below by a softer one $R^{-3}U(R^{-1}\cdot)$, for some $R \gg \ell^{-1}$, which satisfies $\int U = 1$. Similarly, the three-body version of Dyson's lemma \cite{Nam2023condensation} allows us to sacrifice part of the kinetic energy to bound the potential $\ell^2W(\ell\cdot)$ from below by a softer one $R^{-6}\tilde{U}(R^{-1}\cdot)$, which satisfies $\int\tilde{U} = 1$. The key argument in our lower bound is to distinguish between different particle configurations and to apply either lemma depending on that. Let us explain this in more detail. To obtain the leading order of the ground state energy of a system interacting solely via two-body interactions, it is sufficient to consider two particle collisions and to discard the higher-order ones (nonnegativity of the potential is crucial here). In a similar way, for three-body interactions, it is enough to consider only situations where three-particles are nearby and all the others are far away. As one can notice, both situations are essentially mutually exclusive, which allows us to apply either Dyson lemmas depending on the configuration. Hence, we can roughly bound $\tilde{H}_{n,\ell}$ from below by
	\begin{multline*}
		\varepsilon\sum_{i=1}^n-\Delta_i + \ell^{-1} 8\pi a(V)\sum_{1\leq i<j\leq n}R^{-3}U(R^{-1}(x_i-x_j))\\
		+ \ell^{-4}b_\mathcal{M}(W)\sum_{1\leq i<j<k\leq n}R^{-6}\tilde{U}(R^{-1}(x_i-x_j,x_i-x_k)),
	\end{multline*}
	where we kept $\varepsilon$ of the kinetic energy.
	
	\sloppy Then, we use the Temple inequality (see Lemma~\ref{th:temple_inequality}) for the soft potentials $R^{-3}U(R^{-1}\cdot)$ and $R^{-6}\tilde{U}(R^{-1}\cdot)$. To do so, the spectral gap of the kinetic operator must dominate the expected value of the interaction potentials:
	\begin{equation*}
		\varepsilon \gtrsim n^2\ell^{-1}8\pi a(V) + n^3\ell^{-4}b_\mathcal{M}(W),
	\end{equation*}
	which is made possible thanks to the assumption $\ell \ll \ell_\textmd{GP}$. Hence, Temple's inequality yields
	\begin{equation}
		\label{eq:energy_lower_bound_intro}
		\tilde{H}_{n,\ell} \geq \left(n^2\ell^{-1}4\pi a(V) + \dfrac{1}{6}n^3\ell^{-4}b_\mathcal{M}(W)\right)(1 + \mathcal{O}((\rho\mathfrak{a}^3)^\nu)),
	\end{equation}
	for some constant $\nu > 0$. This is the energy on a single box of side length $\ell$ and we need to multiply it by the number of boxes $(L/\ell)^3$ to obtain the energy on the large box $\Lambda_L$. Recalling that $n\simeq\rho \ell^3$ and $H_{n,\ell} = \ell^{-2}\mathcal{U}^*\tilde{H}_{n,\ell}\mathcal{U}$, we obtain the claim.
	
	\medskip
	\noindent
	\textit{Upper bound.} We use a Jastrow factor trial state
	\begin{equation}
		\label{eq:trial_state_upper_bound_strategy}
		\Psi_{N,L} = \prod_{1\leq i<j\leq N}f_{\ell_1}(x_i - x_j)\prod_{1\leq i<j<k\leq N}\tilde{f}_{\ell_2}(x_i - x_j, x_i - x_k).
	\end{equation}
	The function $f_{\ell_1}$ describes two-body correlations up to a distance $\ell_1$ and solves the truncated zero-energy scattering equation
	\begin{equation}
		\label{eq:two_body_scattering_equation_strategy}
		(-\Delta_1 - \Delta_2)f_{\ell_1}(x_1 - x_2) + (Vf_{\ell_1})(x_1 - x_2) = \varepsilon_{\ell_1}(x_1 - x_2),
	\end{equation}
	with $\varepsilon_{\ell_1}$ satisfying
	\begin{equation*}
		\int_{\R^3}\varepsilon_{\ell_1} = 8\pi a(V).
	\end{equation*}
	The function $\tilde{f}_{\ell_2}$ describes three-body correlations up to a distance $\ell_2$ and is solution to
	\begin{multline}
		\label{eq:three_body_scattering_equation_strategy}
		(-\Delta_1 - \Delta_2 - \Delta_3)\tilde{f}_{\ell_2}(x_1 - x_2, x_1 - x_3) + (W\tilde{f}_{\ell_2})(x_1 - x_2, x_1 - x_3)\\
		= \tilde{\varepsilon}_{\ell_2}(x_1 - x_2, x_1 - x_3),
	\end{multline}
	with $\tilde{\varepsilon}_{\ell_2}$ satisfying
	\begin{equation*}
		\int_{\R^6}\tilde{\varepsilon}_{\ell_2} = b_\cM(W).
	\end{equation*}
	To evaluate the energy of the trial state \eqref{eq:trial_state_upper_bound_strategy}, we begin by writing
	\begin{multline*}
		\langle\Psi_{N,L},H_{N,L}\Psi_{N,L}\rangle = N\left\langle\Psi_{N,L},-\Delta_1\Psi_{N,L}\right\rangle + \frac{N(N - 1)}{2}\left\langle\Psi_{N,L},V_{12}\Psi_{N,L}\right\rangle\\
		+ \frac{N(N - 1)(N - 2)}{6}\left\langle\Psi_{N,L},W_{123}\Psi_{N,L}\right\rangle
		\numberthis \label{eq:trial_state_energy_strategy},
	\end{multline*}
	where we used the notations $V_{12} = V(x_1 - x_2)$ and $W_{123} = W(x_1 - x_2, x_1 - x_3)$, and symmetry. For readability's sake we also use the notations
	\begin{equation*}
		f_{ij} = f_{\ell_1}(x_i - x_j) \quad \textmd{and} \quad \tilde{f}_{ijk} = \tilde{f}_{\ell_2}(x_i - x_j, x_i - x_k).
	\end{equation*}
	When computing \eqref{eq:trial_state_energy_strategy}, we get many different terms due to $-\Delta_{x_1}\Psi_{N,L}$ and the numerous products in \eqref{eq:trial_state_upper_bound_strategy}. The two main ones are
	\begin{equation*}
		A =  \dfrac{N(N - 1)}{2}\int_{\R^{3N}}\d{}\mathbf{x}_N\dfrac{2\vert\nabla_1f_{12}\vert^2 + V_{12}f_{12}^2}{f_{12}^2}\prod_{i<j}f_{ij}^2\prod_{i<j<k}\tilde{f}_{ijk}^2
	\end{equation*}
	and
	\begin{equation*}
		B = \dfrac{N(N - 1)(N - 2)}{6}\int_{\R^{3N}}\d{}\mathbf{x}_N\dfrac{2\vert\cM\nabla\tilde{f}_{123}\vert^2 + W_{123}\tilde{f}_{123}^2}{\tilde{f}_{123}^2}\prod_{i<j}f_{ij}^2\prod_{i<j<k}\tilde{f}_{ijk}^2,
	\end{equation*}
	where the matrix $\cM$ is given by \eqref{eq:three_body_symmetry_matrix_intro}. Using \eqref{eq:two_body_scattering_equation_strategy} we extract from $A$ the leading order contribution of $V$ and using \eqref{eq:three_body_scattering_equation_strategy} we extract from $B$ the leading order contribution of $W$. More precisely, using $0 \leq f_{\ell_1} \leq 1$ and \eqref{eq:two_body_scattering_equation_strategy}, we obtain
	\begin{equation}
		\label{eq:trial_state_two_body_contribution_bound}
		A \lesssim 4\pi a(V)N^2\int_{\R^{3(N - 1)}}\d{}x_2\dots\d{}x_N\prod_{2 \leq i< j\leq N}f_{ij}^2\prod_{2\leq i<j<k\leq N}\tilde{f}_{ijk}^2.
	\end{equation}
	Since the state $\Psi_{N,L}$ is \textit{not} normalised, we need to divide $A$ by $\Vert\Psi_{N,L}\Vert^2$ to obtain its actual contribution to the energy. Writing
	\begin{equation}
		\label{eq:trial_state_norm_squared_strategy}
		\Vert\Psi_{N,L}\Vert^2 = \int_{\R^{3N}}\d{}x_1\dots\d{}x_N\prod_{1\leq i<j\leq N}f_{ij}^2\prod_{1\leq i<j<k\leq N}\tilde{f}_{ijk}^2,
	\end{equation}
	we see that this is almost equal to the integral in the right-hand side of \eqref{eq:trial_state_two_body_contribution_bound}, at the exception that there is an extra integral over $x_1$ in \eqref{eq:trial_state_norm_squared_strategy}. To put \eqref{eq:trial_state_norm_squared_strategy} in the correct form, we would like to use the estimate
	\begin{equation}
		\label{eq:trial_state_norm_squared_bad_lower_bound}
		1 - \sum_{j = 2}^N(1 - f_{1j}^2) - \sum_{2\leq j<k\leq N}(1 - \tilde{f}_{1jk}^2) \leq \prod_{j = 2}^Nf_{1j}^2\prod_{2\leq j<k\leq N}\tilde{f}_{1jk}^2,
	\end{equation}
	since, assuming that the two negative terms in the left-hand side are small, this would give us
	\begin{equation}
		\label{eq:trial_state_two_body_contribution_bound_normalised}
		\dfrac{A}{\Vert\Psi_{N,L}\Vert^2} \lesssim 4\pi a(V)\rho N,
	\end{equation}
	with $\rho = N/L^3$. The issue here is that the double sum in \eqref{eq:trial_state_norm_squared_bad_lower_bound} contains too many terms and is thus not small. To circumvent this problem, we use an idea from \cite{Junge2024gse3B}. Namely, thanks to the estimate
	\begin{equation*}
		\tilde{f}_{\ell_2}(x_1 - x_j,x_1 - x_k) \geq \max\left\{\mathds{1}_{\{\vert x_1 - x_j\vert \geq C\ell_2\}},\mathds{1}_{\{\vert x_1 - x_k\vert \geq C\ell_2\}}\right\}
	\end{equation*}
	(see Lemma~\ref{lemma:truncated_three_body_scattering_solution}), we have
	\begin{equation*}
		\prod_{j = 2}^N\mathds{1}_{\{\vert x_1 - x_j\vert \geq C\ell_2\}} \leq \prod_{2\leq j<k\leq N}\tilde{f}_{1jk}^2
	\end{equation*}
	and therefore
	\begin{equation}
		\label{eq:trial_state_norm_squared_good_lower_bound}
		1 - \sum_{j = 2}^N(1 - f_{1j}^2) - \sum_{j = 2}^N\mathds{1}_{\{\vert x_1 - x_j\vert \leq C\ell_2\}} \leq \prod_{j = 2}^Nf_{1j}^2\prod_{2\leq j<k\leq N}\tilde{f}_{1jk}^2.
	\end{equation}
	Compared to \eqref{eq:trial_state_norm_squared_bad_lower_bound}, we traded a double sum for a single sum, which turns out to be sufficient to correctly derive the estimate \eqref{eq:trial_state_two_body_contribution_bound}. Following a similar reasoning for $B$ and for the other terms, we then get
	\begin{equation*}
		\dfrac{\langle\Psi_{N,L},H_{N,L}\Psi_{N,L}\rangle}{\Vert\Psi_{N,L}\Vert^2} \leq N\left(4\pi a(V)\rho + \frac{1}{6}b_\cM(W)\rho^2\right)(1 + \mathcal{O}((\rho\mathfrak{a}^3)^\nu)),
	\end{equation*}
	for some constant $\nu > 0$. Diving both sides by $L^3$ and taking the thermodynamic limit, we obtain the desired claim.
	
	\medskip
	\noindent
	\textbf{Organisation of the paper.} The proof of Theorem~\ref{th:main_result} occupies the rest of this paper. In Section~\ref{section:scattering_energy}, we discuss known facts about the scattering properties of two-body and three-body potentials. Afterwards, in Section~\ref{section:lower_bound} we derive a version of Dyson's lemma that will be used to replace the potentials $V$ and $W$ by softer ones. Then, we use Temple's inequality to prove the energy lower bound in Theorem~\ref{th:main_result}. Finally, the corresponding energy upper bound is proven in Section~\ref{section:energy_upper_bound} by constructing an appropriate trial state.
	
	\medskip
	\noindent
	\textbf{Notations.} We will often omit integration variables, like $\d{}x$, when there is no ambiguity. We will also often use bold variables such as $\mathbf{x}$ to denote variables in $\R^6$.
	
	\medskip
	\noindent
	\textbf{Acknowledgments.} The author would like to express his sincere gratitude to Arnaud Triay and Phan Thành Nam for their continued support and guidance. The author also thanks Julien Ricaud and the anonymous referees for their precious feedback. Partial support by the Deutsche Forschungsgemeinschaft (DFG, German Research Foundation) through the TRR 352 Project ID. 470903074 and by the European Research Council through the ERC CoG RAMBAS  Project Nr. 101044249 is acknowledged.
	
	\section{Scattering energy}\label{sec2}
	
	\label{section:scattering_energy}
	
	In this section we recall some known results on the zero-scattering problem associated to the two-body potential $V$ (see for instance \cite[Appendix C]{Lieb2005MathematicsBG} or \cite[Section 2.2]{Nam2022Bogoliubov}). We also recall some properties of the zero-scattering problem associated to the three-body potential $W$, which were proven in \cite[Section 2.2]{Nam2023condensation}.
	
	\subsection{Truncated two-body scattering solution}
	
	\label{section:truncated_scattering_solution}
	
	Recall that the scattering length $a(V)$ of $V$ is defined by
	\begin{equation}
		\label{eq:scattering_length}
		8\pi a(V) = \inf_{g\in D^1(\mathbb{R}^3)} \int_{\mathbb{R}^3}\d{}x\left(2\left|\nabla g(x)\right|^2 + V(x)\left|1 - g(x)\right|^2\right).
	\end{equation}
	\begin{lemma}[Truncated two-body scattering solution]
		\label{lemma:truncated_two_body_scattering_solution}
		Let $V$ be a two-body potential satisfying Assumption~\ref{assumption:potentials}. Then, the variational problem \eqref{eq:scattering_length} has a unique minimiser $\omega$. Let $0 \leq \chi\leq 1$ be a smooth radial function such that $\chi(x) = 1$ if $|x|\leq 1/2$ and $\chi(x) = 0$ if $|x|\geq 1$. For any $\ell>0$, we define
		\begin{equation*}
			\chi_\ell = \chi(\ell^{-1}\cdot), \quad \omega_{\ell} = \chi_\ell\omega \quad \textmd{and} \quad f_{\ell} = 1 - \omega_{\ell}.
		\end{equation*}
		Then, for any $\ell\geq 2R_0$, the function $f_\ell$ is solution to
		\begin{equation}
			\label{eq:truncated_scattering_equation}
			-2\Delta f_\ell = -Vf_\ell + \varepsilon_{\ell}
		\end{equation}
		on $\R^3, $with $\varepsilon_\ell$ satisfying
		\begin{equation}
			\label{eq:truncated_scattering_solution_error_estimate}
			\vert\varepsilon_\ell(x)\vert \leq Ca(V)\ell^{-3}\mathds{1}_{\left\{\vert x\vert \leq \ell\right\}} \quad \textmd{and} \quad \int_{\R^3}\varepsilon_{\ell} = 8\pi a(V).
		\end{equation}
		Moreover, we have the pointwise estimates
		\begin{equation}
			\label{eq:truncated_scattering_solution_pointwise_estimate}
			\quad 0\leq f_\ell(x) \leq 1, \quad 0 \leq \omega_{\ell}(x) \leq a(V)\dfrac{\mathds{1}_{\{\vert x\vert \leq \ell\}}}{\vert x\vert}, \quad \vert\nabla f_\ell(x)\vert \leq Ca(V)\dfrac{\mathds{1}_{\{\vert x\vert \leq \ell\}}}{\vert x\vert^2}
		\end{equation}
		and
		\begin{equation}
			\label{eq:truncated_scattering_solution_pointwise_estimate_2}
			0 \leq 1 - f_\ell^2(x) \leq Ca(V)\dfrac{\mathds{1}_{\{\vert x\vert \leq \ell\}}}{\vert x\vert},
		\end{equation}
		for all $x\in\mathbb{R}^3$.
	\end{lemma}
	
	\subsection{Three-body scattering problem}
	
	\label{section:modified_scattering_energy}
	
	Let $W:\mathbb{R}^6\rightarrow\mathbb{R}_+$ satisfy the three-body symmetry \eqref{eq:three_body_symmetry} and let us briefly explain why the three-body scattering problem in $\R^9$ associated to $W(x_1 - x_2, x_1 - x_3)$ naturally corresponds to an effective scattering problem in $\R^6$ (see \cite[Section 2.2]{Nam2023condensation} for more detail). We begin by considering the three-body operator
	\begin{equation*}
		-\Delta_{x_1} -\Delta_{x_2} -\Delta_{x_3} + W(x_1 - x_2, x_1 - x_3) \quad \textmd{on $L^2\left(\R^9\right)$}.
	\end{equation*}
	Using the change of variables
	\begin{equation}
		\label{eq:centre_of_mass_removal}
		r_1 = \dfrac{1}{3}\left(x_1 + x_2 + x_3\right), \quad r_2 = x_1 - x_2 \quad \textmd{and} \quad r_3 = x_1 - x_3,
	\end{equation}
	and removing the centre of mass, we obtain the two-body operator
	\begin{equation*}
		-2\Delta_{\mathcal{M}} + W(x,y) \quad \textmd{on $L^2\left(\R^6\right)$},
	\end{equation*}
	where the matrix $\mathcal{M}:\mathbb{R}^3\times\mathbb{R}^3\rightarrow\mathbb{R}^3\times\mathbb{R}^3$ is given by
	\begin{equation}
		\label{eq:three_body_symmetry_matrix}
		\mathcal{M} = \left(\dfrac{1}{2}
		\begin{pmatrix}
			2 & 1\\
			1 & 2
		\end{pmatrix}\right)^{1/2} = \dfrac{1}{2\sqrt{2}}
		\begin{pmatrix}
			\sqrt{3} + 1 & \sqrt{3} - 1\\
			\sqrt{3} - 1 & \sqrt{3} + 1
		\end{pmatrix},
	\end{equation}
	and $-\Delta_{\mathcal{M}}$ is defined as
	\begin{equation*}
		-\Delta_{\mathcal{M}} = \left|\mathcal{M}\nabla_{\mathbb{R}^6}\right|^2 = \diver\left(\mathcal{M}^2\nabla_{\mathbb{R}^6}\right).
	\end{equation*}
	Hence, it is natural to define the \textit{three-body scattering energy}
	\begin{equation}
		\label{eq:modified_scattering_energy}
		b_{\mathcal{M}}(W) = \inf_{g\in D^1(\mathbb{R}^6)}\int_{\mathbb{R}^6}\d{}\mathbf{x}\left(2\left|\mathcal{M}\nabla g(\mathbf{x})\right|^2 + W(\mathbf{x})\left|1 - g(\mathbf{x})\right|^2\right).
	\end{equation}
	\begin{lemma}[Truncated three-body scattering solution]
		\label{lemma:truncated_three_body_scattering_solution}
		Let $W$ be a three-body potential satisfying Assumption~\ref{assumption:potentials}. Then, the variational problem \eqref{eq:modified_scattering_energy} has a unique minimiser $\tilde{\omega}$. Let $0 \leq \chi \leq 1$ be a smooth radial function such $\chi(\mathbf{x}) = 1$ if $\vert\mathbf{x}\vert \leq 1/2$ and $\chi(\mathbf{x}) = 0$ if $\vert\mathbf{x}\vert \geq 1$, and define $\tilde{\chi} \coloneqq \chi(\cM^{-1}\cdot)$. For any $\ell > 0$, we define
		\begin{equation*}
			\tilde{\chi}_\ell = \tilde{\chi}(\ell^{-1}\cdot), \quad \tilde{\omega}_\ell = \tilde{\chi}_\ell\tilde{\omega} \quad \textmd{and} \quad \tilde{f}_\ell = 1 - \tilde{\omega}_\ell,
		\end{equation*}
		which all satisfy the three-body symmetry \eqref{eq:three_body_symmetry}. Then, for any $\ell \geq 2R_0$, the function $\tilde{f}_\ell$ is solution to
		\begin{equation}
			\label{eq:three_body_scattering_equation}
			-\Delta_\cM \tilde{f}_\ell = -W\tilde{f}_\ell + \tilde{\varepsilon}_\ell
		\end{equation}
		on $\R^6$, with $\tilde{\varepsilon}_\ell$ satisfying
		\begin{equation}
			\label{eq:three_body_scattering_solution_error_estimate}
			\vert\tilde{\varepsilon}_\ell(\mathbf{x})\vert \leq Cb_\cM(W)\ell^{-6}\mathds{1}_{\{\vert\mathbf{x}\vert \leq C\ell\}} \quad \textmd{and} \quad \int_{\R^6}\tilde{\varepsilon}_\ell = b_\cM(W).
		\end{equation}
		Moreover, we have the pointwise estimates
		\begin{equation}
			\label{eq:three_body_scattering_solution_pointwise_estimate}
			0 \leq \tilde{f}_\ell(\mathbf{x}) \leq 1, \quad \vert\tilde{\omega}(\mathbf{x})\vert \leq Cb_\cM(W)\dfrac{\mathds{1}_{\{\vert\mathbf{x}\vert \leq C\ell\}}}{\vert\mathbf{x}\vert^4},
		\end{equation}
		\begin{equation}
			\label{eq:three_body_scattering_solution_pointwise_estimate2}
			\vert\nabla\tilde{f}_\ell(\mathbf{x})\vert \leq Cb_\cM(W)\dfrac{\mathds{1}_{\{\vert\mathbf{x}\vert \leq C\ell\}}}{\vert\mathbf{x}\vert^5} \quad \textmd{and} \quad 0 \leq 1 - \tilde{f}_\ell^2(\mathbf{x}) \leq Cb_\cM(W)\dfrac{\mathbf{1}_{\{\vert\mathbf{x}\vert \leq \ell\}}}{\vert\mathbf{x}\vert^4},
		\end{equation}
		for all $\mathbf{x}\in\R^6$. Furthermore, by defining $\tilde{\ell} = \sqrt{3/2}\ell$ and
		\begin{equation*}
			\tilde{g}_\ell(x) = \mathds{1}_{\{\vert x\vert \geq \tilde{\ell}\}},
		\end{equation*}
		we have
		\begin{equation}
			\label{eq:three_body_scattering_solution_lower_bound_disentangling}
			\tilde{f}_\ell(x_1,x_2) \geq \max(\tilde{g}_\ell(x_1),\tilde{g}_\ell(x_2)),
		\end{equation}
		for all $x_1,x_2\in\R^3$.
	\end{lemma}
	
	\begin{proof}
		Up to \eqref{eq:three_body_scattering_solution_pointwise_estimate2} included, everything in Lemma~\ref{lemma:truncated_three_body_scattering_solution} was shown in \cite[Theorem 8]{Nam2023condensation}, or is proven analogously.
		
		The estimate \eqref{eq:three_body_scattering_solution_lower_bound_disentangling} is taken from \cite[Lemma 2]{Junge2024gse3B} and directly follows from the fact that $\tilde{f}_\ell(x_1,x_2) = 1$ when $\vert\cM^{-1}(x_1,x_2)\vert ^{-1} \geq \ell$, which is true whenever $\vert x_1\vert \geq \tilde{\ell}$ or $\vert x_2\vert \geq \tilde{\ell}$.
	\end{proof}

	\section{Lower bound}
	\label{section:lower_bound}
	
	In this section we prove the lower bound
	\begin{equation*}
		\label{eq:ground_state_energy_lower_bound}
		e(\rho,V,W) \geq \left(4\pi a(V)\rho^2 + \dfrac{1}{6}b_\mathcal{M}(W)\rho^3\right)\left(1+\mathcal{O}(Y^\nu)\right) \quad \textmd{when $Y\coloneqq\rho \mathfrak{a}^3\rightarrow 0$},
	\end{equation*}
	where we recall that $\mathfrak{a} = \max(a(V),\rho b_\mathcal{M}(W))$. Here, $\nu > 0$ is a universal constant. We follow the general strategy of \cite{Lieb1998GSE,Nam2016large,Nam2022ground}, that is we estimate the energy in the box $[-L/2,L/2]^3$ by the sum of the energies in smaller boxes $[-\ell/2,\ell/2]^3$ with Neumann boundary conditions. We discard the interactions between the different boxes using the nonnegativity of the potentials $V$ and $W$, and we control the number of particles in each box using a superadditivity argument. The energy in each of these boxes is described by the Hamiltonian $H_{n,\ell}$ defined in \eqref{eq:hamiltonian}, where $n$ denote the number of particles in each box. For convenience, we introduce the unitary operator $\mathcal{U}:L^2_\textmd{s}(\Lambda_\ell^{n}) \rightarrow L^2_\textmd{s}(\Lambda_1^{n})$ that acts as $\mathcal{U}\Psi = \ell^{3n/2}\Psi(\ell\cdot)$, and the rescaled Hamiltonian $\tilde{H}_{n,\ell}=\ell^{2}\mathcal{U}H_{n,\ell}\mathcal{U}^*$. In other words, we consider
	\begin{equation}
		\label{eq:hamiltonian_rescaled}
		\tilde{H}_{n,\ell} = \sum_{i=1}^n-\Delta_i + \sum_{1\leq i<j\leq n}\ell^2V(\ell(x_i-x_j)) + \sum_{1\leq i<j<k\leq n}\ell^2W(\ell(x_i-x_j,x_i-x_k))
	\end{equation}
	acting on $L^2_\s(\Lambda_1^{n})$, where $-\Delta$ denotes the Neumann Laplacian on $L^2(\Lambda)$. In the following proposition we extract the leading order contribution of $\tilde{H}_{n,\ell}$.
	\begin{proposition}[Energy at short length scales]
		\label{prop:energy_short_length}
		Let $V$ and $W$ satisfy Assumption~\ref{assumption:potentials}. Define $\mathfrak{a} = \max(a(V),\rho b_\mathcal{M}(W))$ and $Y = \rho \mathfrak{a}^3$. Let
		\begin{equation*}
			\dfrac{1}{3} < \alpha < \dfrac{12}{35}, \quad \ell\sim \mathfrak{a}Y^{-\alpha} \quad \textmd{and} \quad 0\leq n\leq 10\rho \ell^3.
		\end{equation*}
		Then, the rescaled Hamiltonian defined in \eqref{eq:hamiltonian_rescaled} is such that
		\begin{equation*}
			\displaystyle\tilde{H}_{n,\ell} \geq \dfrac{4\pi a(V)}{\ell}n(n-1) + \dfrac{b_\mathcal{M}(W)}{6\ell^4}n(n-1)(n-2) + \mathcal{O}\left(\rho^2\mathfrak{a}\ell^5Y^\nu\right)
		\end{equation*}
		as $Y \rightarrow 0$, for some constant $\nu>0$, under the assumption that $R_0/\mathfrak{a}$ remains bounded.
	\end{proposition}
	
	To prove Proposition~\ref{prop:energy_short_length} we will replace the singular potentials $\ell^2V(\ell\cdot)$ and $\ell^2W(\ell\cdot)$ by softer ones: $N^{-1}R^{-3}U(R^{-1}\cdot)$ and $N^{-2}R^{-6}\tilde{U}(R^{-1}\cdot)$ respectively. We do so using the Dyson lemma \cite[Lemma 1]{Dyson1957gseHS} (see also \cite[Lemma 1]{Lieb1998GSE}) and a three-body adaptation of it \cite[Lemma 3]{Nam2022ground}. Implementing this at the many-body level is however not trivial and requires proving a suitable Dyson lemma for joint two-body and three-body interactions. We will conclude the proof using Temple's inequality \cite{temple_inequality}.
	
	\subsection{Dyson Lemmas}
	
	\label{subsec:dyson_lemma}
	
	\begin{lemma}[Dyson Lemma for radial potentials]
		\label{lemma:dyson_lemma_radial_potentials}
		Let $V$ be a two-body potential satisfying Assumption~\ref{assumption:potentials}. Let $R_1,R_2 > 0$ such that $R_0 < R_1 < R_2$ and $\{|x| \leq R_2\}\subset\Lambda$. Then, for any nonnegative radial function $U\in C(\Lambda)$ with $\int U = 1$ and $\Supp U\subset\left\{R_1\leq|x|\leq R_2\right\}$, we have the operator inequality
		\begin{equation*}
			-2\nabla_x\mathds{1}_{\left\{|x|\leq R_2\right\}}\nabla_x + V(x) \geq 4\pi a(V)U(x)
		\end{equation*}
		on $L^2(\Lambda)$.
	\end{lemma}
	
	\begin{proof}
		See \cite[Lemma 1]{Lieb1998GSE}.
	\end{proof}
	
	\begin{lemma}[Dyson lemma for three-body potentials]
		\label{lemma:dyson_lemma_three_body_symmetry}
		Let $W$ be a three-body potential satisfying Assumption~\ref{assumption:potentials}. Let $R_1,R_2 > 0$ such that $R_0 < R_1 < R_2$ and $\{|\mathbf{x}|\leq R_2\}\subset\Lambda^2$. Then, for any nonnegative function $\tilde{U}\in C(\Lambda^2)$ that satisfies the three-body symmetry \eqref{eq:three_body_symmetry} with $\Supp \tilde{U} \subset \{R_1\leq|\mathbf{x}|\leq R_2\}$ and $\int\tilde{U} = 1$, we have the operator inequality
		\begin{equation*}
			-2\nabla_{\mathbf{x}}\mathds{1}_{\left\{|\mathbf{x}|\leq R_2\right\}}\nabla_{\mathbf{x}} + W(\mathbf{x}) \geq b_\mathcal{M}(W)\left(1 - \dfrac{CR_0}{R_1}\right)\tilde{U}(\mathbf{x})
		\end{equation*}
		on $L^2_\s(\Lambda^2)$, for some universal constant $C > 0$.
	\end{lemma}
	
	\begin{proof}
		See \cite[Lemma 3]{Nam2022ground}.
	\end{proof}
	
	To prove Proposition~\ref{prop:energy_short_length}, we wish to implement the previous two lemmas at the many-body level. This is done in the next lemma and is one of the main novelties of this work. The key idea that allows us to combine Lemmas~\ref{lemma:dyson_lemma_radial_potentials}~and~\ref{lemma:dyson_lemma_three_body_symmetry} is to separate two spatial configurations and discard the rest. Namely, a given particle either has one close neighbour, two close neighbours or strictly more than two, which is the configuration we neglect (the nonnegativity of the potentials is important here).
	
	We also need to deal with the boundary of the domain. As it is, a particle close to the boundary of $\Lambda$ might interact with particles outside the box. To fix this, we define
	\begin{equation}
		\label{eq:def_box_with_margin}
		\Lambda_\eta = (1-\eta)\Lambda,
	\end{equation}
	for $\eta>0$, where we recall that $\Lambda = [-1/2,1/2]^3$.
	\begin{lemma}[Many-body Dyson lemma with two-body and three-body interactions]
		\label{lemma:dyson_many_body}
		Let $V$ and $W$ satisfy Assumption~\ref{assumption:potentials}. Let $U$ and $\tilde{U}$ be any two functions defined as in Lemmas~\ref{lemma:dyson_lemma_radial_potentials}~and~\ref{lemma:dyson_lemma_three_body_symmetry} respectively, with $\Supp U\subset\{1/4 \leq |x|\leq 1/2\}$ and $\Supp \tilde{U}\subset\{1/4 \leq |\mathbf{x}|\leq 1/2\}$. Denote
		\begin{equation*}
			U_R = R^{-3}U(R^{-1}\cdot), \quad \tilde{U}_R = R^{-6}\tilde{U}(R^{-1}\cdot) \quad \textmd{and} \quad \theta_{2R}(x) = \mathds{1}_{\{\vert x\vert > 2R\}},
		\end{equation*}
		for $R > 0$. Then, for all $\ell,R,\eta>0$ such that $1 > \eta > R/2$ and $R/4 > R_0/\ell$, we have the operator inequality
		\begin{multline*}
			\displaystyle\sum_{i=1}^n-\Delta_{x_i} + \dfrac{1}{2}\sum_{\substack{1\leq i,j\leq n\\i\neq j}}\ell^2V(\ell(x_i-x_j)) + \dfrac{1}{6}\sum_{\substack{1\leq i,j,k\leq n\\i\neq j\neq k\neq i}}\ell^2W(\ell(x_i - x_j, x_i - x_k))\\
			\begin{aligned}
				&\geq \displaystyle\displaystyle \dfrac{4\pi a(V)}{\ell}\sum_{\substack{1\leq i,j\leq n\\i\neq j}}U_R(x_i-x_j)\mathds{1}_{\Lambda_\eta}(x_i)\prod_{\substack{1\leq l\leq n\\ m\neq i,j}}\theta_{2R}\left(\dfrac{x_i+x_j}{2} - x_m\right)\\
				&\phantom{\geq} + \displaystyle \dfrac{b_\mathcal{M}(W)}{6\ell^4}\left(1-C\dfrac{R_0}{\ell R}\right)\sum_{\substack{1\leq i,j,k\leq n\\i\neq j\neq k\neq i}}
				\begin{multlined}[t]
					\tilde{U}_R(x_i-x_j,x_i-x_k)\mathds{1}_{\Lambda_\eta}(x_i)\\
					\times\prod_{\substack{1\leq m\leq n\\ m\neq i,j,k}}\theta_{2R}\left(\dfrac{x_i+x_j+x_k}{3} - x_m\right)
				\end{multlined}
			\end{aligned}
		\end{multline*}
		on $L^2_\s(\Lambda^n)$. Here $C>0$ is a universal constant (independent of $V,W,U,\tilde{U},R,\ell,n$ and $\eta$) and we recall that $\mathfrak{a} = \max(a(V),\rho b_\mathcal{M}(W))$.
	\end{lemma}
	
	\begin{proof}
		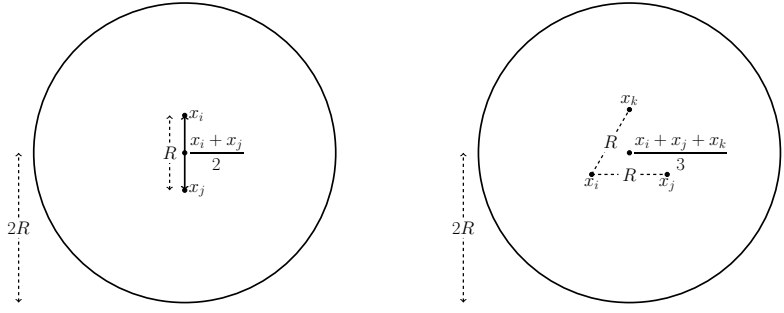
\begin{figure}[ht]
			\centering
			\begin{subfigure}{.4\textwidth}
				\centering
				\resizebox{0.8\textwidth}{!}{%
					\input{fig01}
				}
				\caption{Configuration for which the l.h.s of \eqref{eq:no_three_body_collision_bound} is nonzero. The particle $x_j$ must be a a distance less than $R$ from $x_i$ and all the others must lie outside the black circle.}
				\label{fig:sfig1}
			\end{subfigure}\hskip6ex%
			\begin{subfigure}{.4\textwidth}
				\centering
				\resizebox{0.8\textwidth}{!}{%
					\input{fig02}
				}
				\caption{Configuration for which the l.h.s of \eqref{eq:no_four_body_collision_bound} is nonzero. The particles $x_j$ and $x_k$ must be a a distance less than $R$ from $x_i$ and all the others must lie outside the black circle.}
				\label{fig:sfig2}
			\end{subfigure}
			\caption{Illustration of the two possible configurations for which the l.h.s of \eqref{eq:no_three_and_four_body_collision_bound} is nonzero. Both configurations are incompatible.}
			\label{fig:fig1}
		\end{figure}
		
		First, we denote $\chi_R(x)=\mathds{1}_{\left\{\vert x\vert\leq R\right\}}$ and define
		\begin{equation*}
			\label{eq:cut_off_two_body}
			F_{ij} = \chi_R(x_i-x_j)\prod_{\substack{1\leq m\leq n\\ m\neq i,j}}\theta_{2R}\left(\dfrac{x_i+x_j}{2} - x_m\right),
		\end{equation*}
		for all $\left(x_1,\cdots,x_n\right)\in\left(\mathbb{R}^3\right)^n$.
		This cut-off is such that for every $i\in\{1,\dots, n\}$ there could be at most one $j\neq i$ such that $F_{ij} \neq 0$. To phrase it differently, there exists at most one particle $x_j$ located at a distance smaller than $R$ from $x_i$ such that all the other particles $x_m$ are located at a distance greater than $2R$ from the barycentre of $x_i$ and $x_j$, and thus at a distance greater than $R$ from $x_i$. Hence, we have the "no three-body collision" bound
		\begin{equation}
			\label{eq:no_three_body_collision_bound}
			\sum_{\substack{1\leq j\leq n\\j\neq i}}F_{ij} \leq 1
		\end{equation}
		(see Figure~\ref{fig:sfig1}). Similarly, we define
		\begin{equation*}
			\label{eq:cut_off_three_body}
			\widetilde{F}_{ijk} = \chi_R(x_i - x_j)\chi_R(x_i - x_k)\prod_{\substack{1\leq m\leq n\\ m\neq i,j,k}}\theta_{2R}\left(\dfrac{x_i+x_j+x_k}{3} - x_m\right).
		\end{equation*}
		This cut-off is such that for every $i\in\{1,\dots, n\}$  there could be at most one pair $(j, k)$ with $i\neq j\neq k\neq i$ such that $\widetilde{F}_{ijk} = \widetilde{F}_{ikj} \neq 0$. Again, this means that there exists at most one pair of particles $x_j$ and $x_k$ located at a distance smaller than $R$ from $x_i$ such that all the other particles $x_m$ are located at a distance greater than $2R$ from the barycentre of $x_i$, $x_j$ and $x_k$, and thus at a distance greater than $R$ from $x_i$. Thus we have the "no four-body collision" bound
		\begin{equation}
			\label{eq:no_four_body_collision_bound}
			\sum_{\substack{1\leq j,k\leq n\\i\neq j\neq k\neq i}}\widetilde{F}_{ijk}\leq 2
		\end{equation}
		(see Figure~\ref{fig:sfig2}).	Furthermore, it turns out that, for every $i\in\{1,\dots, n\}$, the left-hand sides of \eqref{eq:no_three_body_collision_bound} and \eqref{eq:no_four_body_collision_bound} cannot both be nonzero at the same time. Indeed, if the left-hand side of \eqref{eq:no_three_body_collision_bound} is nonzero for a given particle $i$, then there exists only one particle located at a distance less than $R$ from the particle $i$ and all the others are located at a distance great than $R$ from $x_i$. Thus, there cannot exist two particles $j$ and $k$ located at a distance less than $R$ from $x_i$, which means that the left-hand side of \eqref{eq:no_four_body_collision_bound} is null. Conversely, if the left-hand side of \eqref{eq:no_four_body_collision_bound} is nonzero for a given particle $i$, then there exist two particles located at a distance less than $R$ from the particle $i$, and the left-hand side of \eqref{eq:no_three_body_collision_bound} must therefore be zero. Hence, the inequalities \eqref{eq:no_three_body_collision_bound} and \eqref{eq:no_four_body_collision_bound} can simply be combined into the inequality
		\begin{equation}
			\label{eq:no_three_and_four_body_collision_bound}
			\sum_{\substack{1\leq j\leq n\\j\neq i}}F_{ij} +
			\dfrac{1}{2}\sum_{\substack{1\leq j,k\leq n\\i\neq j\neq k\neq i}}\widetilde{F}_{ijk} \leq 1,
		\end{equation}
		for every $i\in\{1,\dots, n\}$ (see Figure~\ref{fig:fig1}).	Multiplying the above inequality to the left and to the right by $\mathbf{p}_i \coloneqq -\mathbf{i}\nabla_{x_i}$, then summing over $i$, we obtain
		\begin{align*}
			\sum_{i=1}^n-\Delta_i &\geq \sum_{\substack{1\leq i,j\leq n\\i\neq j}}\mathbf{p}_iF_{ij}\mathbf{p}_i + \dfrac{1}{2}\sum_{\substack{1\leq i,j,k\leq n\\i\neq j\neq k\neq i}}\mathbf{p}_i\widetilde{F}_{ijk}\mathbf{p}_i\\
			&\geq \dfrac{1}{2}\sum_{\substack{1\leq i,j\leq n\\i\neq j}}\sum_{q\in\{i,j\}}\mathbf{p}_qF_{ij}\mathbf{p}_q + \dfrac{1}{6}\sum_{\substack{1\leq i,j,k\leq n\\i\neq j\neq k\neq i}}\sum_{q\in\{i,j,k\}}\mathbf{p}_q\widetilde{F}_{ijk}\mathbf{p}_q.
		\end{align*}
		Hence, we can write
		\begin{multline}
			\label{eq:dyson_many_body_proof_separated_interactions}
			\sum_{i=1}^n-\Delta_i + \dfrac{1}{2}\sum_{\substack{1\leq i,j\leq n\\ j\neq i}}\ell^2V(\ell(x_i-x_j)) + \dfrac{1}{6}\sum_{\substack{1\leq i,j,k\leq n\\i\neq j\neq k\neq i}}\ell^2W(\ell(x_i - x_j,x_i-x_k))\\
			\begin{aligned}[t]
				&\geq \dfrac{1}{2}\sum_{\substack{1\leq i,j\leq n\\i\neq j}}\left[\sum_{q\in\{i,j\}}\mathbf{p}_qF_{ij}\mathbf{p}_q + \ell^2V(\ell(x_i-x_j))\right]\\
				&\phantom{\geq} + \dfrac{1}{6}\sum_{\substack{1\leq i,j,k\leq n\\i\neq j\neq k\neq i}}\left[\sum_{q\in\{i,j,k\}}\mathbf{p}_q\widetilde{F}_{ijk}\mathbf{p}_q + \ell^2W(\ell(x_i-x_j,x_i-x_k))\right].
			\end{aligned}
		\end{multline}
		We now apply Lemma~\ref{lemma:dyson_lemma_radial_potentials} to the first term in the right-hand side of the previous inequality and Lemma~\ref{lemma:dyson_lemma_three_body_symmetry} to the second one. For two-body interactions, we consider the following change of variables
		\begin{equation*}
			c_{ij}=\dfrac{1}{2}\left(x_i + x_j\right),\quad \text{and} \quad r_{ij} = x_i - x_j.
		\end{equation*}
		Then, denoting $\mathbf{p}_x = -\mathbf{i}\nabla_x$, we have
		\begin{equation*}
			\mathbf{p}_i = \dfrac{1}{2}\mathbf{p}_{c_{ij}} + \mathbf{p}_{r_{ij}} \quad \text{and} \quad \mathbf{p}_j = \dfrac{1}{2}\mathbf{p}_{c_{ij}} - \mathbf{p}_{r_{ij}},
		\end{equation*}
		and thus
		\begin{equation*}
			\sum_{q\in\{i,j\}}\mathbf{p}_qF_{ij}\mathbf{p}_q = \dfrac{1}{2}\mathbf{p}_{c_{ij}}F_{ij}\mathbf{p}_{c_{ij}} + 2 \mathbf{p}_{r_{ij}}F_{ij}\mathbf{p}_{r_{ij}}.
		\end{equation*}
		We can ignore the kinetic energy of the centre of mass $\mathbf{p}_{c_{ij}}F_{ij}\mathbf{p}_{c_{ij}}$ for a lower bound since it is nonnegative, which leaves us with
		\begin{equation*}
			\sum_{q\in\{i,j\}}\mathbf{p}_qF_{ij}\mathbf{p}_q \geq  2\mathbf{p}_{r_{ij}}\mathds{1}_{\left\{\vert r_{ij}\vert\leq R\right\}}\mathbf{p}_{r_{ij}}\prod_{\substack{1\leq m\leq n\\ m\neq i,j}}\theta_{2R}(c_{ij} - x_m).
		\end{equation*}
		Moreover, since $V\geq0$, we trivially have the bound
		\begin{equation*}
			V(\ell(x_i - x_j)) \geq V(\ell r_{ij})\mathds{1}_{\Lambda_\eta}(c_{ij})\prod_{\substack{1\leq m\leq n\\ m\neq i,j}}\theta_{2R}(c_{ij} - x_m),
		\end{equation*}
		where we recall that $\Lambda_\eta$ was introduced in \eqref{eq:def_box_with_margin} to neglect the interactions with particles outside the box $\Lambda$. Finally, since $\Supp U_R\subset\left\{R/4\leq|x|\leq R/2\right\} \subset\left\{R_0/\ell\leq|x|\leq\eta\right\}$, we can apply Lemma~\ref{lemma:dyson_lemma_radial_potentials} to obtain
		\begin{multline*}
			\sum_{q\in\{i,j\}}\mathbf{p}_qF_{ij}\mathbf{p}_q + \ell^2V(\ell(x_i - x_j))\\
			\begin{aligned}[b]
				&\geq \left(2\mathbf{p}_{r_{ij}}\mathds{1}_{\left\{|r_{ij}|\leq R\right\}}\mathbf{p}_{r_{ij}} + \ell^2V(\ell r_{ij})\right)\mathds{1}_{\Lambda_\eta}(c_{ij})\prod_{\substack{1\leq m\leq n\\ m\neq i,j}}\theta_{2R}(c_{ij} - x_m)\\
				&\geq \dfrac{4\pi a(V)}{\ell}U_R(r_{ij})\mathds{1}_{\Lambda_\eta}(c_{ij})\prod_{\substack{1\leq m\leq n\\ m\neq i,j}}\theta_{2R}(c_{ij} - x_m).
			\end{aligned} \numberthis \label{eq:dyson_many_body_proof_first_dyson}
		\end{multline*}
		
		For three-body interactions, we do a change of variables similar to \eqref{eq:centre_of_mass_removal}:
		\begin{equation*}
			c_{ijk} = \dfrac{1}{3}\left(x_i + x_j + x_k\right), \quad r_{ij} = x_i - x_j, \quad \text{and} \quad r_{ik} = x_i - x_k.
		\end{equation*}
		Proceeding as before and introducing $\mathbf{r}_{ijk} = (r_{ij},r_{ik})$, we obtain
		\begin{equation*}
			\sum_{q\in\{i,j,k\}}\mathbf{p}_q\widetilde{F}_{ijk}\mathbf{p}_q \geq 2\mathcal{M}\mathbf{p}_{\mathbf{r}_{ijk}}\widetilde{F}_{ijk}\mathcal{M}\mathbf{p}_{\mathbf{r}_{ijk}},
		\end{equation*}
		where we recall that the matrix $\mathcal{M}$ is given by \eqref{eq:three_body_symmetry_matrix}. Using the bound
		\begin{equation*}
			\widetilde{F}_{ijk} = \chi_R(r_j)\chi_R(r_k)\prod_{\substack{1\leq m\leq n\\ m\neq i,j,k}}\theta_{2R}(r_i-x_m) \geq \mathds{1}_{\left\{|\mathbf{r}_{jk}|\leq R/2\right\}}\prod_{\substack{1\leq m\leq n\\ m\neq i,j,k}}\theta_{2R}(r_i-x_m),
		\end{equation*}
		we thus get
		\begin{equation*}
			\sum_{q\in\{i,j,k\}}\mathbf{p}_q\widetilde{F}_{ijk}\mathbf{p}_q \geq 2\mathcal{M}\mathbf{p}_{\mathbf{r}_{ijk}}\mathds{1}_{\left\{|\mathbf{r}_{ijk}|\leq R/2\right\}}\mathcal{M}\mathbf{p}_{\mathbf{r}_{ijk}}\prod_{\substack{1\leq m\leq n\\ m\neq i,j,k}}\theta_{2R}(c_{ijk} - x_m).
		\end{equation*}
		Moreover, from $W\geq 0$ we have the obvious bound
		\begin{equation*}
			W(\ell(x_i-x_j,x_i-x_k)) \geq W(\ell\mathbf{r}_{ijk})\mathds{1}_{\Lambda_\eta}(c_{ijk})\prod_{\substack{1\leq m\leq n\\ m\neq i,j,k}}\theta_{2R}(c_{ijk} - x_m).
		\end{equation*}
		Since $\Supp \tilde{U}_R\subset\{R/4\leq |\mathbf{x}|\leq R/2\}\subset\{R_0/\ell\leq |\mathbf{x}|\leq \eta\}$, we can therefore apply Lemma~\ref{lemma:dyson_lemma_three_body_symmetry} to obtain
		\begin{multline*}
			\sum_{q\in\{i,j,k\}}\mathbf{p}_q\widetilde{F}_{ijk}\mathbf{p}_q + \ell^2W(\ell(x_i - x_j,x_i-x_k))\\
			\begin{aligned}[b]
				&\geq \left[2\mathcal{M}\mathbf{p}_{\mathbf{r}_{ijk}}\mathds{1}_{\left\{|\mathbf{r}_{ijk}|\leq R/2\right\}}\mathcal{M}\mathbf{p}_{\mathbf{r}_{ijk}} + \ell^2W(\ell\mathbf{r}_{ijk})\right]\mathds{1}_{\Lambda_\eta}(c_{ijk})\prod_{\substack{1\leq m\leq n\\ m\neq i,j,k}}\theta_{2R}(c_{ijk} - x_m)\\
				&\geq \left(1-C\dfrac{R_0}{\ell R}\right)\dfrac{b_\mathcal{M}(W)}{\ell^4}U_R(\mathbf{r}_{ijk})\mathds{1}_{\Lambda_\eta}(c_{ijk})\prod_{\substack{1\leq m\leq n\\ m\neq i,j,k}}\theta_{2R}(c_{ijk} - x_m).
			\end{aligned}\numberthis \label{eq:dyson_many_body_proof_second_dyson}
		\end{multline*}
		
		Injecting both \eqref{eq:dyson_many_body_proof_first_dyson} and \eqref{eq:dyson_many_body_proof_second_dyson} into \eqref{eq:dyson_many_body_proof_separated_interactions} concludes the proof of Lemma~\ref{lemma:dyson_many_body}.
	\end{proof}
	
	\subsection{Proof of Proposition~\ref{prop:energy_short_length}: Energy at short length scales}
	
	As the proof of Proposition~\ref{prop:energy_short_length} relies heavily on Temple's inequality \cite{temple_inequality}, we recall it for clarity (see e.g. \cite[Theorem~XIII.5]{methods_modern_mathematical_physics_4}).
	
	\begin{lemma}[Temple's inequality]
		\label{th:temple_inequality}
		Let $A$ be a bounded from below self-adjoint operator with compact resolvent. Let $\lambda_0(A)$ and $\lambda_1(A)$ denote its first two eigenvalues. Then, for any normalised state $\Psi\in D(A)$ and $\gamma < \lambda_1(A)$ such that $\left<A\right>_\Psi\coloneqq \left<\Psi,A\Psi\right> < \gamma$, we have
		\begin{equation}
			\label{eq:temple_inequality}
			\lambda_0(A) \geq \left<A\right>_\Psi - \dfrac{\left<A^2\right>_\Psi - \left<A\right>_\Psi^2}{\gamma - \left<A\right>_\Psi}.
		\end{equation}
	\end{lemma}
	
	Recall that $\mathfrak{a} = \max(b(V),\rho b_\mathcal{M}(W))$ and $Y = \rho\mathfrak{a}^3$, and that we wish to prove that the rescaled Hamiltonian defined in \eqref{eq:hamiltonian_rescaled} is such that
	\begin{equation*}
		\displaystyle\tilde{H}_{n,\ell} \geq \dfrac{b(V)}{2\ell}n(n-1) + \dfrac{b_\mathcal{M}(W)}{6\ell^4}n(n-1)(n-2) + \mathcal{O}\left(\rho^2\mathfrak{a}\ell^5Y^\nu\right)
	\end{equation*}
	when $Y \rightarrow 0$.
	
	\begin{proof}[Proof of Proposition~\ref{prop:energy_short_length}]
		Let $U\in C(\Lambda)$ be nonnegative, radial, and such that $\int U = 1$ and $\Supp U \subset\{1/4\leq |x|\leq 1/2\}$. Let $\tilde{U}\in C(\Lambda^2)$ be nonnegative, satisfying the three-body symmetry \eqref{eq:three_body_symmetry}, and such that $\int \tilde{U}=1$ and $\Supp \tilde{U} \subset\{1/4\leq |\mathbf{x}|\leq 1/2\}$. Let us define $\ell$ and $R$ via
		\begin{equation*}
			\ell = \mathfrak{a}Y^{-\alpha} \quad \textmd{and} \quad R=Y^\beta,
		\end{equation*}
		for some parameters $\beta>\alpha>0$. Let $\eta \coloneqq 2R$ and define
		\begin{equation*}
			\mathcal{T} = \sum_{i=1}^n-\Delta_{x_i},
		\end{equation*}
		\begin{equation*}
			\mathcal{V}_{U} = \dfrac{4\pi a(V)}{\ell}\sum_{\substack{1\leq i,j\leq n\\i\neq j}}U_R(x_i-x_j)\mathds{1}_{\Lambda_\eta}(c_{ij})\prod_{\substack{1\leq m\leq n\\ m\neq i,j}}\theta_{2R}\left(c_{ij} - x_m\right)
		\end{equation*}
		and
		\begin{equation*}
			\mathcal{W}_{\tilde{U}} = \dfrac{b_\mathcal{M}(W)}{6\ell^4}\left(1-C\dfrac{R_0}{\ell R}\right)\sum_{\substack{1\leq i,j,k\leq n\\i\neq j\neq k\neq i}}\tilde{U}_R(x_i-x_j,x_i-x_k)\mathds{1}_{\Lambda_\eta}(c_{ijk})\prod_{\substack{1\leq m\leq n\\m\neq i,j,k}}\theta_{2R}\left(c_{ijk} - x_m\right),
		\end{equation*}
		where $c_{ij} \coloneqq (x_i + x_j)/2$ and $c_{ijk} \coloneqq (x_i + x_j + x_k)/3$. Thanks to the nonnegativity of the potentials $V$ and $W$, and Lemma~\ref{lemma:dyson_many_body} we have
		\begin{equation*}
			\tilde{H}_{n,\ell} \geq \varepsilon\mathcal{T} + (1 - \varepsilon)(\mathcal{V}_U + \mathcal{W}_{\tilde{U}}).
		\end{equation*}
		
		We now apply Temple's inequality to $A \coloneqq \varepsilon\mathcal{T} + (1 - \varepsilon)(\mathcal{V}_U + \mathcal{W}_{\tilde{U}})$, seeing $(1 - \varepsilon)(\mathcal{V}_U + \mathcal{W}_{\tilde{U}})$ as a perturbation of  $\varepsilon\mathcal{T}$. The ground state of $\varepsilon\cT$ is $\Psi_0 = 1$ and it is associated to the eigenvalue $\lambda_0(\varepsilon\cT) = 0$. The second eigenvalue of $\varepsilon\cT$ is $\lambda_1(\varepsilon\cT) = \varepsilon\pi^2/2$. Hence, applying Lemma~\ref{th:temple_inequality} to $A$ with $\gamma = \lambda_1(\varepsilon\mathcal{T})$, and using the nonnegativity of the perturbation $(1 - \varepsilon)(\mathcal{V}_U + \mathcal{W}_{\tilde{U}})$, we obtain
		\begin{equation}
			\tilde{H}_{n,\ell} \geq (1-\varepsilon)\left<\Psi_0,\left(\mathcal{V}_{U} + \mathcal{W}_{\tilde{U}}\right)\Psi_0\right> - 2(1-\varepsilon)^2\dfrac{\left<\Psi_0,\mathcal{V}_{U}^2\Psi_0\right> + \left<\Psi_0,\mathcal{W}_{\tilde{U}}^2\Psi_0\right>}{\varepsilon\pi^2/2 - (1-\varepsilon)\left<\Psi_0,\left(\mathcal{V}_{U} + \mathcal{W}_{\tilde{U}}\right)\Psi_0\right>}, \label{eq:temple_inequality_applied}
		\end{equation}
		as long as
		\begin{equation}
			\label{eq:temple_inequality_applied_condition}
			\lambda_1(\varepsilon\mathcal{T}) - \lambda_0(\varepsilon\mathcal{T}) > (1 - \varepsilon)\left<\Psi_0,(\mathcal{V}_U + \mathcal{W}_{\tilde{U}})\Psi_0\right>.
		\end{equation}
		Here we used the following direct consequence of the Cauchy-Schwartz inequality:
		\begin{equation*}
			\left(\mathcal{V}_{U} + \mathcal{W}_{\tilde{U}}\right)^2 \leq 2\left(\mathcal{V}_{U}^2 + \mathcal{W}_{\tilde{U}}^2\right).
		\end{equation*}
		
		We at once need to control $\left<\Psi_0,\left(\mathcal{V}_U + \mathcal{W}_{\tilde{U}}\right)\Psi_0\right>$, $\left<\Psi_0,\mathcal{V}_U^2\Psi_0\right>$ and $\left<\Psi_0,\mathcal{W}_{\tilde{U}}^2\Psi_0\right>$, and to obtain a precise lower bound on $(1-\varepsilon)\left<\Psi_0,\left(\mathcal{V}_{U} + \mathcal{W}_{\tilde{U}}\right)\Psi_0\right>$. A simple computation yields
		\begin{equation*}
			\left<\Psi_0,\left(\mathcal{V}_{U} + \mathcal{W}_{\tilde{U}}\right)\Psi_0\right> \leq C\left(\dfrac{a(V)}{\ell}n^2 + \dfrac{b_\mathcal{M}(W)}{\ell^4}n^3\right) \leq C Y^{2 - 5\alpha},
		\end{equation*}
		for $Y$ small enough. Thus, taking $\varepsilon \gtrsim Y^{2 - 5\alpha}$ is sufficient to ensure that \eqref{eq:temple_inequality_applied_condition} is valid. Moreover, the no three-body collision bound \eqref{eq:no_three_body_collision_bound} and $\Supp U_R\subset\{|x|\leq R\}$ yield the pointwise bound
		\begin{equation*}
			\label{eq:no_three_body_collision_bound_temple_inequality}
			\displaystyle\sum_{\substack{1\leq i,j\leq n\\i\neq j}}U_R(x_i-x_j)\mathds{1}_{\Lambda_\eta}(c_{ij})\prod_{\substack{1\leq m\leq n\\m\neq i,j}}\theta_{2R}\left(c_{ij} - x_m\right) \leq C\Vert U\Vert_{L^\infty} R^{-3}n,
		\end{equation*}
		from which we obtain
		\begin{equation*}
			\displaystyle\left<\Psi_0,\mathcal{V}_{U}^2\Psi_0\right> \leq \displaystyle C\left(nR^{-3}\dfrac{a(V)}{\ell}\right)^2 \leq  CY^{2-4\alpha-6\beta}.
		\end{equation*}
		Similarly, we have the pointwise bound
		\begin{equation*}
			\label{eq:no_four_body_collision_bound_temple_inequality}
			\displaystyle\sum_{\substack{1\leq i,j,k\leq n\\i\neq j\neq k\neq i}}\tilde{U}_R(x_i-x_j,x_i-x_k)\mathds{1}_{\Lambda_\eta}(c_{ijk})\prod_{\substack{1\leq m\leq n\\m\neq i,j,k}}\theta_{2R}\left(c_{ijk} - x_m\right) \leq C\Vert \tilde{U}\Vert_{L^\infty} R^{-6}n,
		\end{equation*}
		which leads to
		\begin{equation*}
			\displaystyle\left<\Psi_0,\mathcal{W}_{\tilde{U}}^2\Psi_0\right> \leq \displaystyle C\left(nR^{-6}\dfrac{b_\mathcal{M}(W)}{\ell^4}\right)^2 \leq CY^{2\alpha - 12\beta}.
		\end{equation*}
		Hence, by taking $\varepsilon$ such that $\varepsilon\pi^2/4 \gtrsim (1 - \varepsilon)\left<\Psi_0,(\mathcal{V}_U + \mathcal{W}_{\tilde{U}})\Psi_0)\right>$, which is achieved by $\varepsilon \geq CY^{2 -5\alpha}$ for some $C$ large enough, the last term of \eqref{eq:temple_inequality_applied} is estimated by
		\begin{equation}
			\label{eq:energy_short_length_proof1}
			\dfrac{2\left(1-\varepsilon\right)^2\left<\Psi_0\left(\mathcal{V}_U^2 + \mathcal{W}_{\tilde{U}}^2\right)\Psi_0\right>}{\dfrac{\varepsilon\pi^2}{2} - (1-\varepsilon)\left<\Psi_0,\left(\mathcal{V}_{U} + \mathcal{W}_{\tilde{U}}\right)\Psi_0\right>} \leq C\varepsilon^{-1}Y^{2 - 5\alpha}\left(Y^{\alpha-6\beta} + Y^{7\alpha - 2 -12\beta}\right).
		\end{equation}
		We now estimate the second term in \eqref{eq:temple_inequality_applied} from below. Using that
		\begin{equation*}
			\prod_{\substack{1\leq m\leq n\\m\neq i,j}}\theta_{2R}\left(c_{ij} - x_m\right) \geq 1 - \sum_{\substack{1\leq m\leq n\\m\neq i,j}}\left(1-\theta_{2R}\left(c_{ij} - x_m\right)\right),
		\end{equation*}
		we obtain
		\begin{align*}
			\left<\Psi_0,\mathcal{V}_{U}\Psi_0\right> &\geq \dfrac{4\pi a(V)}{\ell}n(n-1)(1-CnR^3)(1-\eta)^3\\
			&\geq \dfrac{4\pi a(V)}{\ell}n(n-1)\left(1-C\left(Y^{1-3(\alpha-\beta)} + Y^{\beta}\right)\right). \numberthis \label{eq:energy_short_length_proof2}
		\end{align*}
		Similarly, using
		\begin{equation*}
			\prod_{\substack{1\leq m\leq n\\m\neq i,j,k}}\theta_{2R}\left(c_{ijk} - x_m\right) \geq 1 - \sum_{\substack{1\leq m\leq n\\m\neq i,j,k}}\left(1 - \theta_{2R}\left(c_{ijk} - x_m\right)\right),
		\end{equation*}
		we obtain
		\begin{align*}
			\left<\Psi_0,\mathcal{W}_{\tilde{U}}\Psi_0\right> &\geq \dfrac{b_{\mathcal{M}}(W)}{6\ell^4}n(n-1)(n-2)\left(1-CnR^3\right)(1-\eta)^3\left(1-\dfrac{CR_0}{\ell R}\right)\\
			&\geq \dfrac{b_\mathcal{M}(W)}{6\ell^4}n(n-1)(n-2)\left(1 - C\left(Y^{1-3(\alpha-\beta)} + Y^\beta + Y^{\alpha-\beta} + Y^{1+3\beta-2\alpha}\right)\right).
			\numberthis \label{eq:energy_short_length_proof3}
		\end{align*}
		Here we used the assumption that $R_0/\mathfrak{a}$ is bounded.
		
		Finally, by combining \eqref{eq:energy_short_length_proof1}, \eqref{eq:energy_short_length_proof2} and \eqref{eq:energy_short_length_proof3}, we obtain
		\begin{align*}
			\tilde{H}_{n,\ell} &\geq \dfrac{4\pi a(V)}{\ell}n(n-1) + \dfrac{b_\mathcal{M}(W)}{6\ell^4}n(n-1)(n-2)\\
			&\phantom{\geq} - CY^{2-5\alpha}\left(Y^{1 - 3(\alpha - \beta)} + Y^\beta + Y^{\alpha - \beta} + Y^{1 + 3\beta - 2\alpha}\right)\\
			&\phantom{\geq} - CY^{2-5\alpha}\left(\varepsilon + \varepsilon^{-1}Y^{\alpha - 6\beta} + \varepsilon^{-1}Y^{7\alpha - 2 - 12\beta}\right).
		\end{align*}
		We take $\alpha$ and $\beta$ to satisfy the conditions
		\begin{equation*}
			\dfrac{1}{3} < \alpha < \dfrac{2}{5} \quad \textmd{and}\quad \alpha - \dfrac{1}{3} < \beta < \dfrac{7\alpha - 2}{12},
		\end{equation*}
		for which the range for $\beta$ is nonempty. With this choice of parameters and choosing $\varepsilon=\left(Y^{\alpha-6\beta} + Y^{7\alpha-2-12\beta}\right)^{1/2}\gg Y^{2-5\alpha}$, we conclude the proof of Proposition~\ref{prop:energy_short_length}.
	\end{proof}
	
	\subsection{Proof of the lower bound in Theorem~\ref{th:main_result}: division into small boxes}
	
	\label{section:proof_energy_lower_bound}
	
	\begin{proof}[Proof of the lower bound in Theorem~\ref{th:main_result}]
		\sloppy We divide the box $\Lambda_L = [-L/2, L/2]^3$ into $M^3$ smaller boxes of side length $\ell$. We parametrise $M(L)=\lfloor LY^\alpha/\mathfrak{a}\rfloor$, for some $\alpha\in(1/3,2/5)$, and we take $\ell(L)>0$ such that $L = M\ell$. Hence,
		\begin{equation*}
			\lim_{\substack{N\rightarrow\infty\\N/L^3\rightarrow\rho}}\ell(L) = \dfrac{\mathfrak{a}}{Y^\alpha}.
		\end{equation*}
		Let $(B_i)_{1\leq i\leq M^3}$ denote these boxes and let us impose Neumann boundary condition on them, as it only lowers the energy. Furthermore, let us use the notation $\overline{A} \coloneqq \mathbb{R}^3\setminus A$, for $A\subset\mathbb{R}^3$. For $\Psi\in L_\textmd{s}^2(\Lambda_L^N)$ and $k\in\{1,\dots,N\}$, we define
		\begin{equation*}
			c_k \coloneqq \dfrac{1}{M^3}\int_{\Lambda_L^N}\sum_{i=1}^{M^3}\delta_{k,\sum_{n=1}^N\mathds{1}_{B_i}(x_n)}\left|\Psi\right|^2 = \dfrac{1}{M^3}
			\begin{pmatrix}
				N\\
				k
			\end{pmatrix}\sum_{i=1}^{M^3}\int_{B_i^k\times\overline{B_i}^{N-k}}\left|\Psi\right|^2,
		\end{equation*}
		where the second equality is obtained by expanding $1 = \otimes_{i=1}^N(\mathds{1}_{B_i} + \mathds{1}_{\overline{B_i}})$ and using the symmetry of $\Psi$.	The sum $\sum_{i=1}^{M^3}\delta_{k,\sum_{n=1}^N\mathds{1}_{B_i}(x_n)}$ is the number of boxes containing exactly $k$ particles for a given configuration $\mathbf{x}_N=(x_1,\dots,x_N)\in\Lambda^N$, and  $M^3c_k$ is thus the expected number of boxes containing exactly $k$ particles. Using that $\int\left|\Psi\right|^2=1$, we deduce
		\begin{equation*}
			\sum_{k=0}^Nc_k = \int_{\Lambda_L^N}\dfrac{1}{M^3}\sum_{i=1}^{M^3}\sum_{k=0}^N\delta_{k,\sum_{n=1}^N\mathds{1}_{B_i(x_n)}}\left|\Psi\right|^2 = \int_{\Lambda_L^N}\left|\Psi\right|^2=1
		\end{equation*}
		and
		\begin{equation*}
			\sum_{k=0}^Nkc_k = \int_{\Lambda_L^N}\dfrac{1}{M^3}\sum_{i=1}^{M^3}\sum_{k=0}^Nk\delta_{k,\sum_{n=1}^N\mathds{1}_{B_i(x_n)}}\left|\Psi\right|^2 = \dfrac{1}{M^3}\int_{\Lambda_L^N}N\left|\Psi\right|^2=\rho \ell^3.
		\end{equation*}
		
		Let us now deal with the energy. For that, we define $E(\ell,k) \coloneqq \inf\sigma(H_{\ell,k})$. Thanks to a straightforward computation relying on
		\begin{equation*}
			1 = \sum_{i = 1}^{M^3}\mathds{1}_{B_i}(x_1), \quad \mathds{1}_{B_i}(x_1) \geq \mathds{1}_{B_i^2}(x_1,x_2) \quad \textmd{and} \quad \mathds{1}_{B_i}(x_1) \geq \mathds{1}_{B_i^3}(x_1,x_2,x_3)
		\end{equation*}
		and the symmetry of $\Psi$, we have
		\begin{align*}
			\langle\Psi,H_{N,L}\Psi\rangle &\geq  \sum_{i=1}^{M^3}N\int_{B_i\times\Lambda^{N-1}}\left|\nabla_1\Psi(\mathbf{x}_N)\right|^2\\
			&\phantom{\geq} + \sum_{i=1}^{M^3}{N \choose 2}\int_{B_i^2\times\Lambda^{N-2}}V(x_1,x_2)\left|\Psi(\mathbf{x}_N)\right|^2\\
			&\phantom{\geq} + \sum_{i=1}^{M^3}{N \choose 3}\int_{B_i^3\times\Lambda^{N-3}}W(x_1,x_2,x_3)\left|\Psi(\mathbf{x}_N)\right|^2\\
			&= \sum_{i=1}^{M^3}\sum_{k=0}^{N}k{N \choose k}\int_{B_i^{k}\times\overline{B_i}^{N-k}}\left|\nabla_1\Psi(\mathbf{x}_N)\right|^2\\
			&\phantom{=} + \sum_{i=1}^{M^3}\sum_{k=0}^{N}{k \choose 2}{N \choose k}\int_{B_i^{k}\times\overline{B_i}^{N-k}}V(x_1,x_2)\left|\Psi(\mathbf{x}_N)\right|^2\\
			&\phantom{=} + \sum_{i=1}^{M^3}\sum_{k=0}^{N}{k \choose 3}{N \choose k}\int_{B_i^{k}\times\overline{B_i}^{N-k}}W(x_1,x_2,x_3)\left|\Psi(\mathbf{x}_N)\right|^2.
		\end{align*}
		Here we made the abuses of notation $V(x_1,x_2)\coloneqq V(x_1-x_2)$ and $W(x_1,x_2,x_3)\coloneqq W(x_1-x_2,x_1-x_3)$, and used that these expressions are symmetric with respect to exchange of variables. In addition, we let $H^{B_i}_k$ denote the translation of $H_{k,\ell}$ to the box $B_i$ and we define the reduced density matrix
		\begin{multline*}
			\gamma_{i,k}(x_1,\dots,x_k;y_1,\dots,y_k)\\
			\coloneqq {N \choose k}\int_{\overline{B_i}^{N-k}}\Psi(x_1,\dots,x_k,z_{k+1},\dots,z_N)\overline{\Psi(y_1,\dots,y_k,z_{k+1},\dots,z_N)}\d{}\mathbf{z}_{N-k},
		\end{multline*}
		with $\mathbf{z}_{N - k} = (z_{k + 1},\dots, z_N)$. Using that
		\begin{multline}
			\label{eq:lower_bound_trace_ineq}
			\begin{aligned}
				{N \choose k}\int_{B_i^k\times\overline{B_i}^{N-k}}\left[k\left\vert\nabla_1\Psi(\mathbf{x}_N)\right\vert^2 + {k \choose 2}V(x_1-x_2)\left|\Psi(\mathbf{x}_N)\right|^2\right]\\
				+ {N \choose k}\int_{B_i^k\times\overline{B_i}^{N-k}}{k \choose 3}W(x_1-x_2,x_1-x_3)\left|\Psi(\mathbf{x}_N)\right|^2
			\end{aligned}\\
			= \Tr_{B_i^k}H^{B_i}_k\gamma_{i,k}\geq E(\ell,k)\Tr_{B^k_i}\gamma_{i,k}
		\end{multline}
		and that $\sum_{i=1}^{M^3}\Tr_{B^k_i}\gamma_{i,k} = M^3c_k$, we find
		\begin{equation}
			\label{eq:energy_lower_bound_divided_boxes}
			\left<\Psi,H_{N,L}\Psi\right> \geq M^3\sum_{k=0}^Nc_kE(\ell,k).
		\end{equation}
		Note that $\Tr_{B_i^k}H_k^{B_i}\gamma_{i,k}$ is meant in the quadratic form sense and the inequality \eqref{eq:lower_bound_trace_ineq} holds because we imposed Neumann boundary conditions on the boxes $B_i$.
		
		We now split the sum in \eqref{eq:energy_lower_bound_divided_boxes} into two pieces: $0\leq k< 10\rho \ell^3$ and $10\rho \ell^3\leq k$. Using Proposition~\ref{prop:energy_short_length}, we have, for all $k< 10\rho \ell^3$, the inequality
		\begin{equation*}
			E(\ell,k) \geq \dfrac{4\pi a(V)}{\ell^3}k(k-1) + \dfrac{b_\mathcal{M}(W)}{6\ell^6}k(k-1)(k-2) - C\rho^2 \mathfrak{a}\ell^3Y^\nu,
		\end{equation*}
		for some $\nu>0$. On the one hand, by denoting $x = \sum_{k<10\rho \ell^3} c_k k\leq \rho \ell^3$, and using the convexity of $t\mapsto t(t-1)$ and of $t\mapsto t(t-1)(t-2)$, we get
		\begin{equation}
			\label{eq:first_bound_energy}
			\displaystyle\sum_{k< 10\rho \ell^3}c_kE(\ell,k) \geq \dfrac{4\pi a(V)}{\ell^3}x(x-1) + \dfrac{b_\mathcal{M}(W)}{6\ell^6}x(x-1)(x-2) - C\rho^2\mathfrak{a}\ell^3Y^\nu.
		\end{equation}
		On the other hand, thanks to the superadditivity property: $E(\ell,k + k')\geq E(\ell,k) + E(\ell,k')$, we have	
		\begin{multline}
			\sum_{k \geq 10\rho \ell^3}c_kE(\ell,k) \geq  E\left(\ell,\left\lfloor10\rho \ell^3\right\rfloor\right)\sum_{k \geq 10\rho \ell^3}c_k\left\lfloor\dfrac{k}{10\rho \ell^3}\right\rfloor\\
			\geq \dfrac{\rho \ell^3 - x}{2}\bigg[\dfrac{4\pi a(V)}{\ell^3}\left(10\rho \ell^3 - 1\right) + \dfrac{b_\mathcal{M}(W)}{6\ell^6}\left(10\rho \ell^3 - 1\right)\left(10\rho \ell^3 - 2\right) - C\rho\mathfrak{a}Y^\nu\bigg].
			\label{eq:second_bound_energy}
		\end{multline}
		By adding \eqref{eq:first_bound_energy} and \eqref{eq:second_bound_energy}, and noticing that the minimum is attained for $x=\rho \ell^3$, we obtain
		\begin{equation*}
			\dfrac{\left<\Psi,H_{N,L}\Psi\right>}{L^3}
			\geq \dfrac{1}{\ell^3}\bigg[\dfrac{4\pi a(V)}{\ell^3}\rho \ell^3\left(\rho \ell^3 - 1\right) + \dfrac{b_\mathcal{M}(W)}{6\ell^6}\rho \ell^3\left(\rho \ell^3 - 1\right)\left(\rho \ell^3 - 2\right) - C\rho^2\mathfrak{a}\ell^3Y^\nu\bigg].
		\end{equation*}
		Here we used that $L = \ell M$. Recall now that we chose $\ell$ such that $\ell\rightarrow\mathfrak{a}Y^{-\alpha}$ in the thermodynamic limit $N\rightarrow\infty$ and $N/L^3\rightarrow\rho$, which implies $\rho \ell^3\rightarrow Y^{1-3\alpha}\gg1$ in the same limit, for $\alpha > 1/3$. Therefore,
		\begin{equation*}
			\lim_{\substack{N\rightarrow\infty\\N/L^3\rightarrow\rho}}\dfrac{\left<\Psi,H_{N,L}\Psi\right>}{L^3} \geq \left(4\pi a(V)\rho^2 + \dfrac{b_\mathcal{M}(W)}{6}\rho^3\right)\left(1 - 3Y^{3\alpha - 1} + \mathcal{O}(Y^\nu)\right)
		\end{equation*}
		as $Y\rightarrow 0$. This concludes the proof of the lower bound in Theorem~\ref{th:main_result}.
	\end{proof}
	
	\section{Upper bound}
	
	\label{section:energy_upper_bound}
	
	In this section we prove the upper bound
	\begin{equation*}
		e(\rho,V,W) \leq \left(\dfrac{1}{2}b(V)\rho^2 + \dfrac{1}{6}b_\mathcal{M}(W)\rho^3\right)(1 + \mathcal{O}(Y^\nu)) \quad \textmd{when $Y = \rho\mathfrak{a}^3\rightarrow 0$},
	\end{equation*}
	for some $\nu > 0$ and where we recall that we introduced the effective scattering length $\mathfrak{a} = \max(b(V),\rho b_\mathcal{M}(W))$.
	
	\begin{proof}[Proof of the upper bound in Theorem~\ref{th:main_result}]
		We use the trial state
		\begin{equation}
			\label{eq:trial_state_upper_bound}
			\Psi_{N,L} = \prod_{1\leq i<j\leq N}f_{\ell_1}(x_i - x_j)\prod_{1\leq i<j<k\leq N}\tilde{f}_{\ell_2}(x_i - x_j,x_i - x_k),
		\end{equation}
		for some parameters $\ell_1,\ell_2$ that will be fixed later and taken such that $a(V) \ll \ell_1 \ll L$ and $b_\cM(W)^{1/4} \ll \ell_2 \ll L$. The function $f_{\ell_1}$ defined in Lemma~\ref{lemma:truncated_two_body_scattering_solution} describes the two-body correlations up to a distance $\ell_1$ and the function $\tilde{f}_{\ell_2}$ defined in Lemma~\ref{lemma:truncated_three_body_scattering_solution} describes the three-body correlations up to a distance $\ell_2$. Trial states of this form have been first used in \cite{Bijl1940lowestWF,Dingle1949ZPE,Jastrwo1955ManyBody} for systems with two-body interactions and are usually referred to as Jastrow factors. Though Dyson \cite{Dyson1957gseHS} worked with a nonsymmetric trial state describing only nearest neighbour interactions, Jastrow factor trial states can be used to derive the correct leading order of a Bose gas with two-body interactions (see for example \cite{Basti2022gseGP}). Moreover, it was recently shown in \cite{Junge2024gse3B} that a Jastrow factor trial state argument can also be used to derive the correct leading order of a Bose gas with three-body interactions. The key argument from \cite{Junge2024gse3B} is to appropriately use the simple estimate \eqref{eq:three_body_scattering_solution_lower_bound_disentangling} to neglect some of the correlations in \eqref{eq:trial_state_upper_bound} that otherwise prevent us from deriving the correct energy upper bound  \footnote{In \cite{Nam2022ground}, the upper bound on the energy is derived using a localisation argument and second quantisation techniques that require a certain regularity of the interaction potential and that in particular do not cover hard-core interactions. Moreover, in earlier versions of the present paper, available at \href{https://arxiv.org/abs/2402.05646v2}{arXiv:2402.05646v2}, that came out before the key idea from \cite{Junge2024gse3B} was known, the upper bound was also derived using a localisation argument and second quantisation techniques and did not cover hard-core interactions.}. For readability's sake we shall write
		\begin{equation*}
			V_{ij} = V(x_i - x_j), \quad W_{ijk} = W(x_i - x_j,x_i - x_k),
		\end{equation*}
		\begin{equation*}
			f_{ij} = f_{\ell_1}(x_i - x_j), \quad \tilde{f}_{ijk} = \tilde{f}_{\ell_2}(x_i - x_j,x_i - x_k),
		\end{equation*}
		\begin{equation*}
			\nabla_if_{ij} = \nabla_{x_i}f(x_i - x_j) \quad \textmd{and} \quad \nabla_i\tilde{f}_{ijk} = \nabla_{x_i}\tilde{f}_{\ell_2}(x_i - x_j,x_i - x_k).
		\end{equation*}
		Moreover, we define
		\begin{equation*}
			F_N = \prod_{1\leq i<j\leq N}f_{ij} \quad \textmd{and} \quad \tilde{F}_N = \prod_{1\leq i<j<k\leq N}\tilde{f}_{ijk}.
		\end{equation*}
		
		To compute the energy of the trial state \eqref{eq:trial_state_upper_bound}, we first notice that
		\begin{equation*}
			\nabla_1\Psi_{N,L}(x_1,\dots,x_N) = \sum_{p = 2}^N\dfrac{\nabla_1f_{1p}}{f_{1p}}F_N\tilde{F}_N + \sum_{2\leq p < q\leq N}\dfrac{\nabla_1\tilde{f}_{1pq}}{\tilde{f}_{1pq}}F_N\tilde{F}_N.
		\end{equation*}
		Using bosonic symmetry, this implies
		\begin{align*}
			\dfrac{\langle\Psi_{N,L},H_{N,L}\Psi_{N,L}\rangle}{\Vert\Psi_{N,L}\Vert^2} &=
			\begin{multlined}[t]
				N\dfrac{\langle\nabla_1\Psi_{N,L},\nabla_1\Psi_{N,L}\rangle}{\Vert\Psi_{N,L}\Vert^2} + \dfrac{N(N - 1)}{2}\dfrac{\langle\Psi_{N,L},V_{12}\Psi_{N,L}\rangle}{\Vert\Psi_{N,L}\Vert^2}\\
				+ \dfrac{N(N - 1)(N - 2)}{6}\dfrac{\langle\Psi_{N,L},W_{123}\Psi_{N,L}\rangle}{\Vert\Psi_{N,L}\Vert^2}
			\end{multlined}\\
			&= \cI_1 + \cI_2 + \cJ_1 + \cJ_2 + \cK_1 + \cK_2 + \cK_3,
		\end{align*}
		with
		\begin{align*}
			\cI_1 &\ceqq \dfrac{N(N - 1)}{2}\dfrac{\displaystyle \int\d{}\mathbf{x}_N\left(\frac{2\vert\nabla_1f_{12}\vert^2}{f_{12}^2} + V_{12}\right)F_N^2\tilde{F}_N^2}{\int\d{}\mathbf{x}_NF_N^2\tilde{F}_N^2}\\
			\cI_2 &\ceqq N(N - 1)(N - 2)\dfrac{\displaystyle \int\d{}\mathbf{x}_N\frac{\nabla_1f_{12}}{f_{12}}\cdot\frac{\nabla_1f_{13}}{f_{13}}F_N^2\tilde{F}_N^2}{\int\d{}\mathbf{x}_NF_N^2\tilde{F}_N^2}\\
			\cJ_1 &\ceqq 2N(N - 1)(N - 2)\dfrac{\displaystyle \int\d{}\mathbf{x}_N\frac{\nabla_1f_{12}}{f_{12}}\cdot\dfrac{\nabla_1\tilde{f}_{123}}{\tilde{f}_{123}}F_N^2\tilde{F}_N^2}{\int\d{}\mathbf{x}_NF_N^2\tilde{F}_N^2}\\
			\cJ_2 &\ceqq N(N - 1)(N - 2)(N - 3)\dfrac{\displaystyle \int\d{}\mathbf{x}_N\frac{\nabla_1f_{12}}{f_{12}}\cdot\dfrac{\nabla_1\tilde{f}_{134}}{\tilde{f}_{134}}F_N^2\tilde{F}_N^2}{\int\d{}\mathbf{x}_NF_N^2\tilde{F}_N^2}\\
			\cK_1 &\ceqq \dfrac{N(N - 1)(N - 2)}{6}\dfrac{\displaystyle \int\d{}\mathbf{x}_N\left(\dfrac{2\vert\mathcal{M}\nabla\tilde{f}_{123}\vert^2}{\tilde{f}_{123}^2} + W_{123}\right)F_N^2\tilde{F}_N^2}{\int\d{}\mathbf{x}_NF_N^2\tilde{F}_N^2}\\
			\cK_2 &\ceqq N(N - 1)(N - 2)(N - 3)\dfrac{\displaystyle \int\d{}\mathbf{x}_N\dfrac{\nabla_1\tilde{f}_{123}}{\tilde{f}_{123}}\cdot\dfrac{\nabla_1\tilde{f}_{124}}{\tilde{f}_{123}}F_N^2\tilde{F}_N^2}{\int\d{}\mathbf{x}_NF_N^2\tilde{F}_N^2}\\
			\cK_3 &\ceqq \dfrac{N(N - 1)(N - 2)(N - 3)(N - 4)}{4}\dfrac{\displaystyle \int\d{}\mathbf{x}_N\dfrac{\nabla_1\tilde{f}_{123}}{\tilde{f}_{123}}\cdot\dfrac{\nabla_1\tilde{f}_{145}}{\tilde{f}_{145}}F_N^2\tilde{F}_N^2}{\int\d{}\mathbf{x}_NF_N^2\tilde{F}_N^2}.
		\end{align*}
		To rewrite $\cK_1$ we used the identity
		\begin{multline*}
			\vert\nabla_{x_1}\tilde{f}_{\ell_2}(x_1 - x_2,x_1 - x_3)\vert^2 + \vert\nabla_{x_2}\tilde{f}_{\ell_2}(x_1 - x_2,x_1 - x_3)\vert^2\\
			+ \vert\nabla_{x_3}\tilde{f}_{\ell_2}(x_1 - x_2,x_1 - x_3)\vert^2 = 2\vert(\cM\nabla_{\R^6}\tilde{f}_{\ell_2})(x_1 - x_2,x_1 - x_3)\vert^2.
		\end{multline*}
		We only provide the details of the bounds for $\cI_1,\cI_2,\cJ_1,\cK_1$ since the others bounds follow similarly. Thanks to the estimate \eqref{eq:three_body_scattering_solution_lower_bound_disentangling}, we have
		\begin{equation*}
			\prod_{2 \leq j<k\leq N}\tilde{f}_{\ell_2}(x_1 - x_j,x_1 - x_k)^2 \geq \prod_{j = 2}^Ng_{\ell_2}(x_1 - x_j),
		\end{equation*}
		where $\tilde{g}_{\ell_2}(x) = \mathds{1}_{\{\vert x\vert \geq C\ell_2\}}$. Hence, by defining $u_{\ell_1} = 1 - f_{\ell_1}^2$ and $\tilde{v}_{\ell_2} = 1 - g_{\ell_2}$, we have
		\begin{equation*}
			1 - \sum_{j = 2}^Nu_{1j} - \sum_{j = 2}^N\tilde{v}_{1j} \leq \prod_{j = 2}^Nf_{1j}^2\prod_{2\leq j< k\leq N}\tilde{f}_{1jk}^2 \leq 1,
		\end{equation*}
		where we used the short-hand notations $u_{ij} = u_{\ell_1}(x_i - x_j)$ and $\tilde{v}_{ij} = \tilde{v}_{\ell_2}(x_i - x_j)$. This allows us to decouple the variable $x_1$ in the numerator and in the denominator of $\cI_1$ as follows:
		\begin{equation}
			\label{eq:upper_bound_I1}
			\cI_1 \leq 
			\dfrac{N^2}{2}\dfrac{\int\d{}x\left(2\vert\nabla f_{\ell_1}(x)\vert^2 + (Vf_{\ell_1}^2)(x)\right)}{L^3 - CN\int\d{}x(u_{\ell_1}(x)) - CN\int\d{}x(\tilde{v}_{\ell_2}(x))}.
		\end{equation}
		Using an integration by parts, that $f_{\ell_1}$ is solution to \eqref{eq:truncated_scattering_equation}, and the estimates \eqref{eq:truncated_scattering_solution_error_estimate} and \eqref{eq:truncated_scattering_solution_pointwise_estimate}, we have
		\begin{align*}
			\int_{\R^3}\left(2\vert\nabla f_{\ell_1}\vert^2 + (Vf_{\ell_1}^2)\right) = \int_{\R^3}\varepsilon_{\ell_1}f_{\ell_1} = \int_{\R^3}\varepsilon_{\ell_1} - \int_{\R^3}\varepsilon_{\ell_1}\omega_{\ell_1} \leq 8\pi a(V) + C\dfrac{a(V)^2}{\ell_1}.
		\end{align*}
		Combining this with \eqref{eq:upper_bound_I1} and the estimates \eqref{eq:truncated_scattering_solution_pointwise_estimate_2} and \eqref{eq:three_body_scattering_solution_pointwise_estimate2}, we obtain
		\begin{equation*}
			\cI_1 \leq \dfrac{4\pi a(V) \rho N(1 + Ca(V)/\ell_1)}{1 - C\rho a(V)\ell_1^2 - C\rho\ell_2^3} \leq 4\pi a(V)\rho N\left[1 + C\dfrac{a(V)}{\ell_1} +  C\rho a(V)\ell_1^2 + C\rho\ell_2^3\right]
		\end{equation*}
		under the assumption that $a(V)/\ell_1,\rho a(V)\ell_1^2,\rho\ell_2^3 \ll 1$.
		
		To bound $\cI_2$, we similarly wish to decouple the variables $x_1$ and $x_2$ in the numerator and in the denominator. For this, we define $\tilde{u}_{\ell_2} \coloneq 1 - \tilde{f}_{\ell_2}$, which is such that
		\begin{multline}
			\label{eq:estimate_numerator_denominator_decoupling}
			1 - \sum_{j = 2}^Nu_{1j} - \sum_{j=3}^Nu_{2j} - \sum_{j=3}^N\tilde{v}_{1j} - \sum_{j = 3}^N\tilde{v}_{2j} - \sum_{k = 3}^N\tilde{u}_{12k}\\
			\leq \prod_{j=2}^Nf_{1j}^2\prod_{j=3}^Nf_{2j}^2\prod_{3\leq j<k\leq N}\tilde{f}_{1jk}^2\tilde{f}_{2jk}^2\prod_{k = 3}^N\tilde{f}_{12k}^2 \leq 1.
		\end{multline}
		Combining this with the estimates \eqref{eq:truncated_scattering_solution_pointwise_estimate}, \eqref{eq:truncated_scattering_solution_pointwise_estimate_2} and \eqref{eq:three_body_scattering_solution_pointwise_estimate2}, we obtain
		\begin{align*}
			\cI_2 &\leq N^3\dfrac{\left[\int\d{}x\vert\nabla f_{\ell_1}(x)\vert\right]^2}{L^6 - CNL^3\int \d{}x(u_{\ell_1}(x)) - CL^3N\int\d{}x(\tilde{v}_{\ell_2}(x)) - CN\int\d{}\mathbf{x}(\tilde{u}_{\ell_2}(\mathbf{x}))}\\
			&\leq \dfrac{C\rho^2 Na(V)^2\ell_1^2}{1 - C\rho a(V)\ell_1^2 - C\rho\ell_2^3 - C\rho b_\cM(W)\ell_2^2/L^3}\\
			&\leq C\rho^2Na(V)^2\ell_1^2\left[1 + C\rho a(V)\ell_1^2 + C\rho\ell_2^3\right]
		\end{align*}
		when $\rho a(V)\ell_1^2,\rho\ell_2^3 \ll 1$.
		
		We now bound $\cK_1$. Thanks to \eqref{eq:estimate_numerator_denominator_decoupling}, we have
		\begin{equation*}
			\cK_1 \leq \dfrac{N^3}{6}\dfrac{\int\d{}\mathbf{x}\left(2\vert(\cM\nabla\tilde{f})(\mathbf{x})\vert^2 + (W\tilde{f}^2)(\mathbf{x})\right)}{L^6 - CNL^3\int\d{}x(u_{\ell_1}(x)) - CNL^3\int\d{}x(\tilde{v}_{\ell_2}(x)) - CN\int\d{}\mathbf{x}(\tilde{u}_{\ell_3}(\mathbf{x}))}.
		\end{equation*}
		Then, using that $\tilde{f}_{\ell_2}$ solves \eqref{eq:three_body_scattering_equation} and the estimates \eqref{eq:truncated_scattering_solution_pointwise_estimate_2}, \eqref{eq:three_body_scattering_solution_error_estimate} and\eqref{eq:three_body_scattering_solution_pointwise_estimate}, we obtain
		\begin{align*}
			\cK_1 &\leq \dfrac{b_\cM(W)\rho^2N}{6}\dfrac{1 + Cb_\cM(W)/\ell_2^4}{1 - C\rho a(V)\ell_1^2 - C\rho\ell_2^3 - C\rho b_\cM(W)\ell_2^2/L^3}\\
			&\leq \dfrac{1}{6}b_\cM(W)\rho^2N\left[1 + C\rho a(V)\ell_1^2 + C\dfrac{b_\cM(W)}{\ell_2^4} + C\rho\ell_2^3\right],
		\end{align*}
		under the condition that $\rho a(V)\ell_1^2,\rho\ell_2^3 \ll 1$.
		
		To bound $\cJ_1$, we use \eqref{eq:estimate_numerator_denominator_decoupling} with appropriate modifications to decouple the variables $x_1$ and $x_3$ in the numerator and in the denominator. For the numerator, we need to evaluate the double integral
		\begin{equation*}
			\int_{\R^6}\d{}x_1\d{}x_3\vert\nabla_{x_1}f_{\ell_1}(x_1 - x_2)\vert\vert\nabla_{x_1}\tilde{f}_{\ell_2}(x_1 - x_2,x_1 - x_3)\vert
		\end{equation*}
		at fixed $x_2$. By definition of the cut-off at distance $\ell_2$, we can write
		\begin{equation*}
			\nabla_{x_1}\tilde{f}_{\ell_2}(x_1 - x_2,x_1 - x_3) = \mathds{1}_{\{\vert x_2 - x_3\vert \leq C\ell_2\}}\nabla_{x_1}\tilde{f}_{\ell_2}(x_1 - x_2,x_1 - x_3).
		\end{equation*}
		Doing so and using the Cauchy--Schwarz inequality, we are able to integrate the variable $x_3$ and obtain
		\begin{multline*}
			\int_{\R^6}\d{}x_1\d{}x_3\vert\nabla_{x_1}f_{\ell_1}(x_1 - x_2)\vert\vert\nabla_{x_1}\tilde{f}_{\ell_2}(x_1 - x_2,x_1 - x_3)\vert \leq C\eta \ell_2^3\int_{\R^3}\d{}x_1\vert\nabla_{x_1}f_{\ell_1}(x_1 - x_2)\vert^2\\
			+ \dfrac{1}{2\eta}\int_{\R^6}\d{}x_1\d{}x_3\vert\nabla_{x_1}\tilde{f}_{\ell_2}(x_1 - x_2,x_1 - x_3)\vert^2,
		\end{multline*}
		for all $\eta > 0$. Using integrations by parts, the scattering equations \eqref{eq:truncated_scattering_equation} and \eqref{eq:three_body_scattering_equation}, and the estimates \eqref{eq:truncated_scattering_solution_error_estimate} and \eqref{eq:three_body_scattering_solution_error_estimate} as we did when bounding $\cI_1$ and $\cK_1$, we further get
		\begin{equation*}
			\int_{\R^6}\d{}x_1\d{}x_3\vert\nabla_{x_1}f_{\ell_1}(x_1 - x_2)\vert\vert\nabla_{x_1}\tilde{f}_{\ell_2}(x_1 - x_2,x_1 - x_3)\vert \leq \eta \ell_2^3a(V) + C\eta^{-1}b_\cM(W),
		\end{equation*}
		for all $\eta > 0$. Therefore,
		\begin{equation*}
			\cJ_1 \leq CN\rho\mathfrak{a}(\rho\ell_2^3)^{1/2}\left[1 + \rho a(V)\ell_1^2 + \rho\ell_2^3\right].
		\end{equation*}
		
		We can show in a similar way that
		\begin{equation*}
			\cJ_2,\cK_2,\cK_3 \leq  CN\rho\mathfrak{a}\left[1 + \rho a(V)\ell_1^2 + \rho\ell_2^3\right]
		\end{equation*}
		
		Summing up, we have proven that
		\begin{multline*}
			\dfrac{\langle\Psi_{N,L},H_{N,L}\Psi_{N,L}\rangle}{\Vert\Psi_{N,L}\Vert^2} \leq N\left(4\pi a(V)\rho + \dfrac{1}{6}b_\cM(W)\rho^2\right)\\
			\times \left(1 + C\rho a(V)\ell_1^2 + C\dfrac{a(V)}{\ell_1} + C\rho\ell_2^3 + C\dfrac{b_\cM(W)}{\ell_2^4}\right).
		\end{multline*}
		Choosing $\ell_1 = \rho^{-1/3}$ and $\ell_2 = b_\cM(W)^{1/4}(\rho b_\cM(W)^{3/4})^{-1/7}$, dividing by $L^3$ and taking the thermodynamic limit concludes the proof of the upper bound in Theorem~\ref{th:main_result}.
	\end{proof}
	
	\printbibliography

\end{document}

%% file: fig01.tex
\begin{tikzpicture}
	\draw[color=black, ultra thick](0,0) circle (5);
	\filldraw[black] (0,1.25) circle (2pt) node[anchor=west]{\LARGE $x_i$};
	\filldraw[black] (0,-1.25) circle (2pt) node[anchor=west]{\LARGE $x_j$};
	\filldraw[black] (0,0) circle (2pt) node[anchor=west]{\LARGE $\dfrac{x_i + x_j}{2}$};
	\draw[color=black, ultra thick, <->](0,-1.25) -- (0,1.25);
	\node(a) at (-0.5,0) {\LARGE $R$};
	\draw[color=black, very thick, dashed,<-](-0.5,-1.25) -- (a);
	\draw[color=black, very thick, dashed,->](a) -- (-0.5,1.25);
	\node(b) at (-5.5,-2.5) {\LARGE $2R$};
	\draw[color=black, very thick, dashed,<-](-5.5,-5) -- (b);
	\draw[color=black, very thick, dashed,->](b) -- (-5.5,0);
\end{tikzpicture}

%% file: fig02.tex
\begin{tikzpicture}
	\draw[color=black, ultra thick](0,0) circle (5);
	\filldraw[black] (0,{2.5*sqrt(3)/3}) circle (2pt) node[anchor=south]{\LARGE $x_k$};
	\filldraw[black] (-1.25,{-2.5*sqrt(3)/6}) circle (2pt) node[anchor=north]{\LARGE $x_i$};
	\filldraw[black] (1.25,{-2.5*sqrt(3)/6}) circle (2pt) node[anchor=north]{\LARGE $x_j$};
	\filldraw[black] (0,0) circle (2pt) node[anchor=west]{\LARGE $\dfrac{x_i + x_j + x_k}{3}$};
	
	\node(a) at (0,{-2.5*sqrt(3)/6}) {\LARGE $R$};
	\draw[color=black, very thick, dashed](-1.25,{-2.5*sqrt(3)/6}) -- (a);
	\draw[color=black, very thick, dashed](a) -- (1.25,{-2.5*sqrt(3)/6});
	
	\node(c) at ({-2.5/4},{2.5*sqrt(3)/12}) {\LARGE $R$};
	\draw[color=black, very thick, dashed](-1.25,{-2.5*sqrt(3)/6}) -- (c);
	\draw[color=black, very thick, dashed](c) -- (0,{2.5*sqrt(3)/3});	
	

		
	\node(b) at (-5.5,-2.5) {\LARGE $2R$};
	\draw[color=black, very thick, dashed,<-](-5.5,-5) -- (b);
	\draw[color=black, very thick, dashed,->](b) -- (-5.5,0);
\end{tikzpicture}

%% file: bibliography.bib
@article{Adami2023microscopicDS,
	title={Microscopic derivation of a {S}chr\"odinger equation in dimension one with a nonlinear point interaction}, 
	author={Riccardo Adami and Jinyeop Lee},
	journal={J. Funct. Anal.},
	volume={288},
	number={10},
	pages={110866},
	year={2025},
}

@article{Basti2022gseGP,
	title={Ground state energy of a {B}ose gas in the {G}ross--{P}itaevskii regime},
	volume={63},
	ISSN={1089-7658},
	number={4},
	journal={J. Math. Phys.},
	publisher={AIP Publishing},
	author={Basti, Giulia and Cenatiempo, Serena and Olgiati, Alessandro and Pasqualetti, Giulio and Schlein, Benjamin},
	year={2022},
	month=apr
}

@article{Basti2021secondOU,
	author={Basti, G. and Cenatiempo, C. and Schlein, B.},
	title={A new second order upper bound for the ground state energy of dilute {B}ose gases},
	journal={Forum Math. Sigma},
	volume={9},
	pages={E74},
	year={2021},
}

@article{Bijl1940lowestWF,
	title = {The lowest wave function of the symmetrical many particles system},
	journal = {Physica},
	volume = {7},
	number = {9},
	pages = {869-886},
	year = {1940},
	issn = {0031-8914},
	author = {A. Bijl},
}

@article{Chen2011quinticNLS,
	title = {The quintic {NLS} as the mean field limit of a boson gas with three-body interactions},
	journal = {J. Funct. Anal.},
	volume = {260},
	number = {4},
	pages = {959--997},
	year = {2011},
	issn = {0022-1236},
	author = {Thomas Chen and Nataša Pavlović}
}

@article{Chen2012secondOC,
	title={Second Order Corrections to Mean Field Evolution for Weakly Interacting Bosons in The Case of Three-body Interactions},
	volume={203},
	number={2},
	journal={Arch. Ration. Mech. Anal.},
	publisher={Springer Science and Business Media LLC},
	author={Chen, Xuwen},
	year={2012},
	pages={455–497},
}

@article{Chen2018TheDO,
	title={The derivation of the {$\mathbb{T}^{3}$} energy-critical {NLS} from quantum many-body dynamics},
	author={Xuwen Chen and Justin Holmer},
	journal={Invent. Math.},
	year={2019},
	pages={433--547},
	volume={217}
}

@article{Dyson1957gseHS,
	author={Dyson, F. J.},
	title={Ground state energy of a hard-sphere gas},
	journal={Phys. Rev.},
	volume={106},
	year={1957},
	pages={20--26},
}

@article{Dingle1949ZPE,
	author={Dingle, R.},
	title={The zero-point energy of a system of particles},
	journal={Philos. Mag.},
	volume={40},
	year={1949},
	number={304},
	pages={573--578}
}

@article{Fournais2020energyDB,
	author={Fournais, S. and Solovej, J.P.},
	title={The energy of dilute {B}ose gases},
	journal={Ann. of Math.},
	volume={192},
	pages={893--976},
	year={2020},
}

@article{Fournais2023energyDB2,
	author={Fournais, S. and Solovej, J.P.},
	title={The energy of dilute {B}ose gases {II}: the general case},
	journal={Invent. math.},
	volume={232},
	pages={863--994},
	year={2023},
}

@article{Gammal2000atomicBEC,
	year = {2000},
	volume = {33},
	number = {19},
	pages = {4053},
	author = {Gammal, A. and Frederico, T. and Tomio, L. and Chomaz, Ph.},
	title = {Atomic {B}ose--{E}instein
	condensation with three-body interactions and collective excitations},
	journal = {J. Phys. B},
}

@article{Fournais2022groundSE,
	title={The Ground State Energy of a Two-Dimensional {B}ose Gas}, 
	author={S. Fournais and T. Girardot and L. Junge and L. Morin and M. Olivieri},
	journal={Commun. Math. Phys.},
	year={2024},
	volume={405},
	number={59}
}

@misc{Haberberger2023freeED,
	title={The free energy of dilute {B}ose gases at low temperatures}, 
	author={Florian Haberberger and Christian Hainzl and Phan Thành Nam and Robert Seiringer and Arnaud Triay},
	year={2023},
	eprint={2304.02405},
	archivePrefix={arXiv},
	primaryClass={math-ph},
	url={https://arxiv.org/abs/2304.02405}, 
}

@article{Jastrwo1955ManyBody,
	title = {Many-Body Problem with Strong Forces},
	author = {Jastrow, Robert},
	journal = {Phys. Rev.},
	volume = {98},
	issue = {5},
	pages = {1479--1484},
	numpages = {0},
	year = {1955},
	month = {Jun},
	publisher = {American Physical Society},
}

@misc{Junge2024gse3B,
	title={Ground state energy of a dilute Bose gas with three-body hard-core interactions}, 
	author={Lukas Junge and François Louis Antoine Visconti},
	year={2024},
	eprint={2406.09019},
	archivePrefix={arXiv},
	primaryClass={math-ph},
	url={https://arxiv.org/abs/2406.09019}, 
}

@article{Koch2008stabilisation,
	title={Stabilization of a purely dipolar quantum gas against collapse},
	volume={4},
	number={3},
	journal={Nat. Phys.},
	publisher={Springer Science and Business Media LLC},
	author={Koch, T. and Lahaye, T. and Metz, J. and Fröhlich, B. and Griesmaier, A. and Pfau, T.},
	year={2008},
	pages={218–222},
}

@misc{lee2021rateCT,
	title={Rate of convergence towards mean-field evolution for weakly interacting bosons with singular three-body interactions}, 
	author={Jinyeop Lee},
	year={2021},
	eprint={2006.13040},
	archivePrefix={arXiv},
	primaryClass={math-ph},
	url={https://arxiv.org/abs/2006.13040},
}

@article{Lee1957eigenvaluesEB,
	author={Lee, T. D. and Huang, K. and Yang, C. N.},
	title={Eigenvalues and eigenfunctions of a Bose system of hard spheres and its low temperature properties},
	journal={Phys. Rev.},
	volume={106},
	pages={1135--1145},
	year={1957},
}

@article{Li2021derivationNS,
	author = {Li, Yongsheng and Yao, Fangyan},
	title = "{Derivation of the nonlinear {S}chrödinger equation with a general nonlinearity and {G}ross--{P}itaevskii hierarchy in one and two dimensions}",
	journal = {J. Math. Phys. },
	volume = {62},
	number = {2},
	pages = {021505},
	year = {2021},
}

@book{lieb_loss,
	author={Lieb, E. H. and Loss, M.},
	title={{A}nalysis},
	series={{G}raduate {S}tudies in {M}athematics},
	publisher={{A}merican {M}athematical {S}ociety, second ed.},
	address={Providence},
	volume={14},
	year={2001},
}

@book{Lieb2005MathematicsBG,
	author={Lieb, E. H. and Seiringer, R. and Sobolev, J. P. and Yngvason, J.},
	title={The mathematics of the {B}ose gas and its condensation: {S}eries: {O}berwolfach {S}emin.},
	address={Basel},
	publisher={Birkh{\"a}user {V}erlag},
	volume={34},
	year={2005},
}

@article{Lieb1998GSE,
	author={Lieb, E. H. and Jakob Yngvason},
	title={Ground state energy of the low density {B}ose gas},
	publisher={American {P}hysical {S}ociety ({APS})},
	volume={80},
	number={12},
	year={1998},
	journal={Phys. Rev. Lett.},
}

@article{Nam2023condensation,
	author={Nam, Phan T. and Ricaud, J. and Triay, A.},
	title={The condensation of a trapped dilute {B}ose gas with three-body interactions},
	journal={Prob. Math. Phys},
	volume = {4},
	pages={91--149},
	year={2023},
}

@article{Nam2022dilute,
	author={Nam, Phan T. and Ricaud, J. and Triay, A.},
	title={Dilute {B}ose gas with three-body interaction: recent results and open questions},
	journal={J. Math. Phys.},
	volume={63},
	number={6},
	year={2022},
}

@article{Nam2022ground,
	author={Nam, Phan T. and Ricaud, J. and Triay, A.},
	title={Ground state energy of the low density {B}ose gas with three-body interactions},
	journal={J. Math. Phys.},
	volume={63},
	year={2022},
	pages={071903},
	note={Special collection in honor of
	Freeman Dyson},
}

@article{Nam2016large,
	title={Ground states of large bosonic systems: the {G}ross--{P}itaevskii limit revisited},
	author={Nam, Phan T. and Rougerie, N. and R. Seiringer},
	year={2016},
	journal={Anal. PDE},
	volume={9},
	pages={459--485},
}

@article{Nam2019derivation3D,
	title={Derivation of {3D} energy-critical nonlinear {S}chr\"odinger equation and {B}ogoliubov excitations for {B}ose gases}, 
	author={Phan Thành Nam and Robert Salzmann},
	journal={Commun. Math. Phys.},
	year={2019},
	volume={375},
	pages={495--571},
}

@article{Nam2022Bogoliubov,
	author={Nam, Phan T. and Triay, A.},
	title={{B}ogoliubov excitation spectrum of trapped {B}ose gas in the {G}ross--{P}itaevskii regime},
	journal={J. Math. Pures Appl.},
	volume={176},
	pages={18--101},
	year={2022},
}

@article{Nguyen2023onedimensional,
	author = {Nguyen, Dinh-Thi and Ricaud, Julien},
	title = {On One-Dimensional {B}ose Gases with Two\nobreakdash-Body and (Critical) Attractive Three\nobreakdash-Body Interactions},
	fjournal = {SIAM Journal on Mathematical Analysis},
	journal = {SIAM J.~Math. Anal.},
	publisher = {SIAM},
	coden = {SJMAAH},
	issn = {0036-1410; 1095-7154},
	volume = {56},
	number = {3},
	pages = {3203--3251},
	year = {2024},
}

@article{Nguyen2023stabilization,
	title = {Stabilization against collapse of {2D} attractive {B}ose--{E}instein condensates with repulsive, three\nobreakdash-body interactions}, 
	author={Dinh-Thi Nguyen and Julien Ricaud},
	year={2025},
	journal={Lett. Math. Phys.},
	volume={115},
	number={31},
}

@misc{petrov2023beyondmeanfield,
	title={Beyond-mean-field effects in mixtures: few-body and many-body aspects}, 
	author={D. S. Petrov},
	year={2023},
	eprint={2312.05336},
	archivePrefix={arXiv},
	primaryClass={cond-mat.quant-gas},
	url={https://arxiv.org/abs/2312.05336},
}

@article{surface_liquid_water_three_body,
	author={Pieniazek, P. A. and Tainter, C. J. and Skinner, J. L.},
	title={Surface of liquid water: {T}hree-body
	interactions and vibrational sum-frequency spectroscopy},
	journal={J. Am. Chem. Soc.},
	volume={133},
	pages={10360--10363},
	year={2011},
}

@book{methods_modern_mathematical_physics_4,
	author={Reed, M. C. and Simon, B.},
	title={Methods of modern mathematical physics},
	address={New York},
	publisher = {Academic Press},
	volume={4: Analysis of operators},
	year={1978},
}

@article{Rout2024microscopicDG,
	title={A microscopic derivation of {G}ibbs measures for the {1D} focusing quintic nonlinear {S}chr\"{o}dinger equation}, 
	author={Andrew Rout and Vedran Sohinger},
	year={2023},
	journal={Commun. Partial Differ. Equ.},
	pages={1008--1055},
	volume={48}
}

@book{statistical_mechanics,
	author={Ruelle, D.},
	title={Statistical mechanics. {R}igorous results},
	year={1999},
	publisher={Singapore: World Scientific. London: Imperial College Press},
	address={London}
}

@inproceedings{temple_inequality,
	author={Temple, G.},
	title={The theory of {R}ayleigh's principle as applied to continuous systems},
	booktitle={Proc. Roy. Soc. London A},
	volume={119},
	pages={276},
	year={1928}
}

@article{three_body_interactions_condensed_phases_helium,
	author={Ujevic, S. and Vitiello, S. A.},
	title={Three-body interactions in the condensed phases of helium
	atom systems},
	journal={J. Phys. Condens. Matter},
	volume={19},
	pages={116212},
	year={2007},
}

@article{Yau2009secondOU,
	author={Yau, H.-T and Yin, J.},
	title={The second order upper bound for the ground state energy of a {B}ose gas},
	journal={J. Stat. Phys.},
	volume={136},
	pages={453--503},
	year={2009},
}

@article{Yuan2015derivationQNLS,
	title = {Derivation of the Quintic {NLS} from many-body quantum dynamics in {$T^2$}},
	author = {Jianjun Yuan},
	journal = {Commun. Pure Appl. Anal.},
	volume = {14},
	number = {5},
	pages = {1941--1960},
	year = {2015},
}
